\documentclass[conference,compsoc]{IEEEtran}
\ifCLASSOPTIONcompsoc
  \usepackage[nocompress]{cite}
\else
  \usepackage{cite}
\fi
\ifCLASSINFOpdf
\else
\fi
\usepackage[conf={restate,no link to proof}]{proof-at-the-end}
\usepackage{comment}
\usepackage{amsthm}
\usepackage{mathtools}
\usepackage{tcolorbox}
\usepackage{listings}
\lstdefinestyle{sql}{
    language=SQL,
    basicstyle=\ttfamily,
    keywordstyle=\color{blue},
    commentstyle=\color{gray},
    stringstyle=\color{purple},
    showstringspaces=false,
    breaklines=true,
    captionpos=b,
    frame=lines,
    numbers=left,
    numberstyle=\tiny,
    numbersep=5pt
}
\usepackage{subcaption}
 \usepackage{amsfonts}
\usepackage{xcolor}
\usepackage[linesnumbered,ruled,vlined]{algorithm2e}
\usepackage{etoolbox}
\usepackage{makecell}
\usepackage{balance}
\usepackage[T1]{fontenc}
\usepackage{graphicx}
\newbool{ARXIV}
\boolfalse{ARXIV}
\usepackage{enumitem}
\usepackage{float}
\usepackage{color} 
\usepackage{xargs} 
\usepackage{xspace} 
\usepackage{adjustbox}

\def\notationcolor{black} 

\newcommand{\notation}[2]{\newcommand{#1}{{\textcolor{\notationcolor}{\ensuremath{#2}}}}}

\newcommand{\term}[2]{\newcommand{#1}{\textcolor{\notationcolor}{#2}\xspace}}

\notation{\data}{\mathcal{D}}
\notation{\mech}{\mathcal{M}} 
\notation{\datavec}{\vec{x}} 
\notation{\covar}{\Sigma} 
\notation{\bmat}{\mathbf{B}} 
\notation{\comvar}{\covar_{*}} 
\notation{\comb}{\bmat_{*}} 
\notation{\commech}{\mech_{*}} 
\notation{\amat}{\mathbf{A}} 
\notation{\pcostmat}{\mathbf{X}}
\notation{\nullspace}{\mathbf{H}}
\notation{\dimsize}{d}
\notation{\query}{\vec{q}}
\notation{\qmat}{\mathbf{Q}}
\notation{\numquery}{m}
\notation{\senstwo}{\Delta_2}
\notation{\brank}{k}
\notation{\randalg}{\mathcal{A}}
\notation{\loewner}{\preceq}
\notation{\rloewner}{\succeq}
\notation{\rowspace}{\text{rowspace}}
\notation{\trace}{\text{trace}}

\notation{\impression}{\mathbf{i}}
\notation{\conversion}{\mathbf{c}}
\notation{\metadata}{\mathbf{\phi}}

\notation{\threshold}{\mathbf{T}}
\notation{\outp}{\omega}

\term{\obj}{\textbf{OBJ}}

\term{\myalg}{CM}

\notation{\adj}{\mathcal{N}}
\notation{\subeventadj}{\mathcal{N}^{\text{Event}}}
\notation{\subglobaladj}{\mathcal{N}^{\text{Global}}}
\notation{\subperdayadj}{\mathcal{N}^{\text{User}}}

\notation{\diffset}{\mathcal{D}}
\notation{\subeventdiff}{\mathcal{D}^{\text{Event}}}
\notation{\subglobaldiff}{\mathcal{D}^{\text{Global}}}
\notation{\subperdaydiff}{\mathcal{D}^{\text{User}}}

\notation{\diffex}{\mathcal{D}^{\text{User}}_{ex}}

\notation{\wmat}{\mathbf{A}} 
\notation{\pmat}{\mathbf{X}} 
\notation{\dvec}{\mathbf{u}} 

\notation{\bound}{r}
\notation{\thres}{\tau}

\notation{\varfun}{\text{Var}}
\notation{\pcost}{\text{budget}}
\notation{\loss}{\mathcal{L}}
\notation{\abovet}{\text{NumAbove}}
\notation{\gap}{s}

\usepackage{todonotes}

\usepackage{tikz}
\usetikzlibrary{positioning}
\usepackage{pgfplots}
\usepackage{algorithmic}
\usepackage{amsmath}
\newtheoremstyle{mystyle}
  {}
  {}
  {\itshape}
  {}
  {\scshape}
  {.}
  { }
  {}
\theoremstyle{mystyle}

\newtheorem{theorem}{Theorem}[section]

\newtheorem{definition}[theorem]{Definition}

\newtheorem{lemma}[theorem]{Lemma}

\let\vec\mathbf  
\let\mat\mathbf  

\AtBeginDocument{%
  \providecommand\BibTeX{{%
    \normalfont B\kern-0.5em{\scshape i\kern-0.25em b}\kern-0.8em\TeX}}}

\def\matmech{AdsBPC\xspace}
\def\adsmech{AdsBPC\xspace}
\def\bintree{BIN\xspace}
\def\bincrop{Stream\xspace}
\def\idglobal{IPA\xspace}

\def\matglobal{UMM\xspace}
\def\mmbpc{MMBPC\xspace}

\newcommand{\addcolor}[1]{\textcolor{black}{#1}}

\begin{document}


\title{Click Without  Compromise:  Online Advertising Measurement \\ via Per User Differential Privacy}

\makeatletter
\newcommand{\linebreakand}{%
  \end{@IEEEauthorhalign}
  \hfill\mbox{}\par
  \mbox{}\hfill\begin{@IEEEauthorhalign}
}
\makeatother


\author{
    \IEEEauthorblockN{Yingtai Xiao\IEEEauthorrefmark{1}, 
    Jian Du\IEEEauthorrefmark{1}, 
    Shikun Zhang\IEEEauthorrefmark{1}, 
    Wanrong Zhang\IEEEauthorrefmark{1},
    Qiang Yan\IEEEauthorrefmark{1}, 
    Danfeng Zhang\IEEEauthorrefmark{2},
    Daniel Kifer\IEEEauthorrefmark{3}}
    \IEEEauthorblockA{
        \IEEEauthorrefmark{1}TikTok Inc.\;
        \IEEEauthorrefmark{2}Duke University\; \IEEEauthorrefmark{3}Penn State University\;
        \\ 
        \IEEEauthorrefmark{1}\{yingtai.xiao1, jian.du, shikun.zhang, wanrong.zhang, yanqiang.mr\}@tiktok.com
        \\
        \IEEEauthorrefmark{2}danfeng.zhang@duke.edu  \; \IEEEauthorrefmark{3}duk17@psu.edu
    }
}

\maketitle


\pagestyle{empty}

\begin{abstract}
    Online advertising is a cornerstone of the Internet ecosystem, with advertising measurement playing a crucial role in optimizing efficiency. 
Ad measurement entails attributing desired behaviors, such as purchases, to ad exposures across various platforms, necessitating the collection of user activities across these platforms. 
As this practice faces increasing restrictions due to rising privacy concerns, safeguarding user privacy in this context is imperative. 
Our work is the first to formulate the real-world challenge of advertising measurement systems with real-time reporting of streaming data in advertising campaigns. 
We introduce AdsBPC, a novel user-level differential privacy protection scheme for online advertising measurement results. 
This approach optimizes global noise power and results in a non-identically distributed noise distribution that preserves differential privacy while enhancing measurement accuracy. 
Through experiments on both real-world advertising campaigns and synthetic datasets, AdsBPC achieves a 33\% to 95\% increase in accuracy over existing streaming DP mechanisms applied to advertising measurement. This highlights our method's effectiveness in achieving superior accuracy alongside a formal privacy guarantee, thereby advancing the state-of-the-art in privacy-preserving advertising measurement.
\end{abstract}


\section{Introduction}\label{sec:intro}

Online advertising is essential in the Internet ecosystem as it fuels the economy of the web and enables businesses to connect with their target audience effectively. Advertising measurement systems \cite{yun2020challenges, guha2010challenges, briggs1997advertising, yuan2013real} are pivotal in the online advertising landscape. They aim to assess the effectiveness of advertising campaigns by determining the total number of conversions---such as purchases, sign-ups, downloads, or other desired actions---that occur as a result of ad exposure. To reach a more diverse audience, advertisers often set up ad campaigns across various platforms, including Google Ads, Facebook Ads, and other digital advertising networks. To effectively gauge the impact of these ads on each platform, tracking mechanisms are in place to monitor advertising campaigns across the selected channels.

The ad measurement process encompasses two major steps: attribution and computation.
The attribution step involves analyzing the customer's online journey, starting from initial ad exposure to the ultimate conversion. Attribution aims to identify which exposures (touchpoints) contributed to conversions. 
The method that allocates credit to various advertising channels or touchpoints throughout the customer journey is referred to as multi-touch attribution (MTA). MTA recognizes the multiple interactions a customer may have with different advertising channels before making a conversion. By distributing credit among these touchpoints based on their influence on the conversion, MTA models offer a more nuanced understanding of each channel's effectiveness within the advertising mix.
The computation step is to calculate the total conversion value based on the attribution results to gain insights into the effectiveness of different advertising channels.
Based on the such insights, advertisers can improve their advertising campaigns by reallocating budget, adjusting targeting parameters, refining ad designs, or tweaking other campaign elements to improve their return on investment and users' advertising experience.

Ad measurement requires collecting a variety of user data, including the ads users have viewed, clicks registered by the ad platform, and conversion data from the advertiser's end. Traditionally, this involves automatically tracking user behavior across both ad publishers' platforms and advertisers' platforms, utilizing tools such as pixels, tags, or software development kits (SDKs). However,the tracking and collection of user data across sites have raised increasing privacy concerns regarding user consent, data security, and the potential misuse of personal information~\cite{toubiana2010adnostic, yun2020challenges, haddadi2011targeted, ullah2020privacy}. These privacy concerns have led to legislative actions, such as the General Data Protection Regulation (GDPR) in Europe \cite{voigt2017eu} and industry initiatives such as Apple's App Tracking Transparency~\cite{AppleATT}, which impose varying degrees of restrictions on cross-site tracking. Recent studies have developed privacy-preserving mechanisms to support advertising measurement~\cite{ayala2022show, pfeiffer2021masked, rosales2012post}. However, these approaches lack formal privacy guarantees. Our work addresses this gap by introducing a formal differential privacy guarantee for advertising measurement outcomes.

Differential privacy~\cite{dwork06Calibrating,zcdp} provides a mathematical framework for adding carefully calibrated noise to query results or datasets to protect \textit{individual data record} while still allowing for meaningful data analysis. The calibrated noise makes it nearly impossible to discern whether the dataset contains a specific data record or not, thereby safeguarding the privacy of individual data record.
Differential privacy has become a key component in modern privacy protection efforts in research areas such as database systems~\cite{near2021differential}, machine learning model training data protection~\cite{abadi2016deep}, and census statistics~\cite{abowd2018uscensus}.

Employing a differentially private mechanism to protect individual user privacy in a practical MTA ad measurement setting presents two challenges. First, maintaining a continuous data stream poses significant hurdles for the effective implementation of differential privacy. In advertising campaigns, measurement results must be continuously disclosed to both advertisers and ad providers in real-time, per time slot (e.g., daily or hourly). This is crucial as it allows advertisers to observe the outcomes of their campaigns and adjust their strategies accordingly. Unlike one-time data releases in database queries, continuous streams gradually consume the differential privacy budget, leading to significant accuracy deterioration and making the analysis results less valuable from a business standpoint. \addcolor{Existing differentially private streaming algorithms \cite{chan2011private, dwork2010differential, dong2023continual} primarily focus on protecting prefix-sum queries (i.e., releasing the cumulative sum up to the current day) and minimizing the root mean squared error. These algorithms are limited in the types of supported queries and objective functions.}

Second, MTA involves records from various platforms, making traditional \textit{event-level} differential privacy insufficient. In differential privacy, the granularity of privacy refers to the level at which DP guarantees are provided, indicating the unit of data protected by DP. Event-level DP protects whether a specific user action (such as clicking an ad or making a purchase) is included in the dataset, whereas \textit{user-level} DP provides privacy protection at the level of individual users, covering all event-level records such as views, clicks, and conversions of a particular user.
Consequently, user-level DP protects whether a user engaged with the ad campaign at any point, offering a stronger and more meaningful level of privacy protection than event-level DP. For instance, protecting whether a user has ever engaged with an ad campaign for diabetic medication is more meaningful than merely protecting whether a user has clicked on an Instagram ad for diabetic medication.

To address the two challenges above, we propose the \textbf{AdsBPC} algorithm, which offers greater flexibility by supporting various types of online queries and objective functions. This is achieved without significant modifications to the algorithm’s structure—only the initialization of the noise scale parameters needs to be adjusted. The algorithm can then adaptively produce the desired results based on the provided data. Our AdsBPC algorithm is also designed to provide a formal user-level differential privacy. The primary contributions of this paper are as follows:
\begin{enumerate}
\item  We are the first to introduce and formalize the problem regarding \addcolor{online} MTA advertising measurement systems, including both attribution and computation steps, which require the release of measurement results on a streaming basis.
\item We devised a user-level streaming differentially private algorithm, AdsBPC, tailored to address the challenges in MTA advertising measurement. Mathematically, we formulate the noise calibration process as a noise variance minimization problem with privacy budget constraints.
\item \addcolor{We developed a differentially private algorithm designed to accurately bound per-day user contributions. This approach reduces the sensitivity of AdsBPC, thereby enhancing accuracy while maintaining a fixed privacy budget.}
\item Through comprehensive experiments with both synthetic and real-world datasets, our algorithm achieves notable improvements in accuracy, \addcolor{ranging from 33\% to 95\%}, compared to previous works. Our results indicate that our algorithm can provide more accurate and meaningful insights for advertisers while ensuring users’ privacy.

\end{enumerate}

This paper is structured as follows: Section \ref{sec:background} illustrates how advertising measurement operates through concrete examples and discuss related concepts in differential privacy. Section \ref{sec:related} explores previous work in the differential privacy area. Section \ref{sec:problemdef} formally defines our user-level differentially private advertising measurement system and presents an efficient algorithm to address it. Section \ref{sec:discuss} presents discussions on the proposed method. Finally, Section \ref{sec:experiments} presents the results of comprehensive experiments conducted on both synthetic and real-world datasets using our algorithm.

\


\section{Ad Measurement and Differential Privacy }\label{sec:background}

In this section, we elucidate advertising measurement with multi-touch attribution as well as differential privacy definitions in the advertising measurement context.
Mathematical notations used throughout this paper are summarized in Table~\ref{tab:notation}.
\begin{table}[h]
\begin{center}
\caption{Table of Notation}\label{tab:notation}
\resizebox{1 \linewidth}{!}{
\begin{tabular}{|cl|}\hline
$n$: & Number of days. \\
$\datavec$: & Data vector representation of a stream. \\
$\mathcal{N}$: &All pairs of neighboring streams. $\datavec$\\
$\diffset$: & All possible deltas between neighboring. $\datavec$, $\tilde{\datavec}$. \\
$\vec{d}$: & The difference between two neighboring dataset. \\
$\impression_{u, j}$: &$j$-th impression of user $u$.\\
$\conversion_{u, j}$: &$j$-th conversion of user $u$.\\
$w$: &Attribution weight. \\
\hline
$\query_j$: & The $j$-th query vector.\\
$\mat{Q}$: & Query workload matrix.\\
$\mat{L}$: & Linear transformation matrix. \\
$\mat{B}$: & Basis matrix (strategy matrix).\\
$\vec{z}$: & Random variable. \\
$\rho$: & Privacy budget. \\
\hline
$r_i$: &User contribution bound on day $i$. \\
$\sigma_i$: &Noise scale on day $i$.\\
$l $: & Number of days using DP-quantile.\\
$p$: & The quantile percentage used in the DP-quantile.\\
$s^{\uparrow}, s^{\downarrow}$: & Scaling parameters.\\
$ T^{\uparrow}, T^{\downarrow}$: &Thresholds used in the sparse vector technique.\\
\hline
\end{tabular}
}
\end{center}
\end{table}







\subsection{An Example of Advertising Measurement}

At the core of advertising measurement lies the concept of impression and conversion.

\noindent\textbf{Impression Events.} An impression event encompasses several key elements that define the context of an ad being served to a user. These elements include:
1) User ID, 2) Publisher ID, 3) Advertiser ID, 
4) Timestamp,  and 5) Metadata (supplementary information related to the impression, such as the type of user interaction (click or view) and ad format, providing the context for further analysis). We use the notation $\impression_{u, j}$ to express the $j$-th impression of user $u$, and use $\metadata(\impression_{u, j})$ to express the metadata of the impression $\impression_{u, j}$. In the case where an advertiser has multiple products, we use different Advertiser IDs for different products. Table \ref{tab:impression} shows an example of the impression table, where each row represents an impression event of users interacting (viewing or clicking) with ads on various publishers.

\begin{table}[h]
\centering
\caption{Impression Table I on Day 1}
\begin{tabular}{|c|c|c|c|c|c|c|}
\hline
User-ID & Pub-ID & Ad-ID &Imp. &Time & Metadata\\
\hline
$u_1$ & P-1  & Ad-1& $\impression_{u_1, 1}$ & 1 & $\metadata(\impression_{u_1, 1})$ \\
$u_1$ & P-1  & Ad-1& $\impression_{u_1, 2}$ & 11 & $\metadata(\impression_{u_1, 2})$ \\
\hline
$u_2$ & P-1 & Ad-1& $\impression_{u_2, 1}$ & 15 &$\metadata(\impression_{u_2, 1})$ \\
$u_2$ & P-2  & Ad-1& $\impression_{u_2, 2}$  & 25 &$\metadata(\impression_{u_2, 2})$ \\
\hline
\end{tabular}
\label{tab:impression}
\end{table}

\begin{table}[h]
\centering
\caption{Conversion Table C on Day 1}
\begin{tabular}{|c|c|c|c|c|c|}
\hline
User-ID & Ad-ID&  Conversion &Time & Metadata \\
\hline
$u_1$  &Ad-1 & $\conversion_{u_1,1}$&10 & $\metadata(\conversion_{u_1,1})$ \\
\hline
$u_2$  &Ad-1 & $\conversion_{u_2,1}$&20 & $\metadata(\conversion_{u_2,1})$\\
\hline
$u_2$  &Ad-1 & $\conversion_{u_2, 2}$&30 &$\metadata(\conversion_{u_2, 2})$\\
\hline
\end{tabular}
\label{tab:conversion}
\end{table}


\noindent\textbf{Conversion Events}
Conversion events represent pivotal moments in the advertising pipeline when users generate value for advertisers after interacting with ads on the publishers' platforms. These events encompass 1) User ID, 
2) Advertiser ID, 3) Timestamp,  and
4) Metadata: Additional data that characterize the conversion, including its type (e.g., purchase or sign-up) and value (e.g., purchase amount). We use the notation $\conversion_{u, j}$ to express the $j$-th conversion of user $u$, and use $\metadata(\conversion_{u, j})$ to express the metadata of the conversion $\conversion_{u, j}$. Table \ref{tab:conversion} shows an example of the conversion table, where each row represents a conversion, such as a purchase of Nike shoes.




\noindent\textbf{Advertising Measurement Process} The gist of ad measurement can be distilled into a two-step process. In the first step, conversions are attributed to one or more impressions using predefined logic or attribution models. This step is referred to as ``attribution,'' determining which impressions played a role in driving the eventual conversion event. Section \ref{subsec:attribution} gives more details about attribution logic.
Subsequently, in the second step, queries are executed on the resulting attributed dataset to derive valuable aggregate statistics. These queries may involve counting the total number of conversions attributed to a subset of impressions, calculating the corresponding conversion rate, or assessing the total return on an ad.


\begin{table}
\centering
\caption{Table after Joining Impression Table I and Conversion Table C for Day 1 based on User-ID and Ad-ID}
\begin{tabular}{|c|c|c|c|c|c|c|}
\hline
User-ID & Pub-ID & Ad-ID& Imp. & Conv. &I.Time & C.Time\\
\hline
$u_1$ & P-1& Ad-1 & $\impression_{u_1, 1}$ &  $\conversion_{u_1, 1}$ & 1 & 10\\
\hline
$u_2$ & P-1& Ad-1 &  $\impression_{u_2, 1}$ & $\conversion_{u_2, 1}$ & 15 & 20\\
$u_2$ & P-1& Ad-1 & $\impression_{u_2, 1}$ & $\conversion_{u_2, 2}$ & 15 & 30\\
$u_2$ & P-2& Ad-1 & $\impression_{u_2, 2}$ & $\conversion_{u_2, 2}$ & 25 & 30\\
\hline
\end{tabular}
\label{tab:join}
\end{table}

\subsubsection{Attribution Logic}\label{subsec:attribution}

For attribution, we assume there is a trusted third party that joins the impression table and the conversion table. Listing \ref{lst:sql} shows a simplified SQL query that joins these two tables. Table \ref{tab:join} shows the result of joining Tables 2 and 3. Then the third party assigns the attribution weight of a conversion to each relevant impression. The following are several commonly used attribution logic. Last-Touch Attribution (\textbf{LTA}), also known as Last-Click Attribution, assigns full credit for a conversion to the last interaction or touchpoint that directly led to the conversion. First-Touch Attribution (\textbf{FTA}), as the name suggests, attributes full credit for a conversion to the first interaction or touchpoint that introduced the customer to the brand or campaign. Uniform Attribution (\textbf{UNI}) is an attribution model in advertising measurement that evenly distributes credit for conversions across all touchpoints in a customer's journey, providing equal weight to each interaction. 

\begin{lstlisting}[style=sql, caption={SQL Query to Join Impression Table and Conversion Table on Day 1}, label=lst:sql]
SELECT User-ID, Pub-ID, Ad-ID,  
Impression, Conversion
FROM I JOIN C
ON I.User-ID = C.User-ID
AND I.Ad-ID = C.Ad-ID
And I.Time < C.Time
\end{lstlisting}

Table \ref{tab:attribution} shows an example of the above attribution logics, given the joined Table \ref{tab:join}. 
Since only $\impression_{u_2, 1}$ is relevant to $\conversion_{u_2, 1}$, both LTA, FTA, and UNI will assign a weight of 1 to $\impression_{u_2, 1}$ (similarly for $\conversion_{u_1, 1}$). In the case of $\conversion_{u_2, 2}$, where both $\impression_{u_2, 1}$ and $\impression_{u_2, 2}$ are relevant, LTA attributes a weight of 1 to $\impression_{u_2, 2}$ as it occurs last. FTA assigns a weight of 1 to $\impression_{u_2, 1}$ since it happens first. UNI, on the other hand, evenly distributes the weight, assigning 0.5 to both $\impression_{u_2, 1}$ and $\impression_{u_2, 2}$.

\begin{table}
\centering
\caption{Attribution on Day 1 in Table~\ref{tab:join}}
\begin{tabular}{c|| c|c || c|c ||c| c  }
\hline
Conversion &\multicolumn{2}{c||}{$\conversion_{u_1, 1}$} &\multicolumn{2}{c||}{$\conversion_{u_2, 1}$} &\multicolumn{2}{c}{$\conversion_{u_2, 2}$} \\
\hline
 Pub-ID & P-1 & P-2 & P-1 &  P-2 & P-1 & P-2  \\
\hline
\hline
LTA & 1& 0  & 1& 0 & 0&  1  \\
\hline
FTA  & 1& 0  & 1& 0& 1&  0  \\
\hline
UNI   &1 & 0   & 1& 0 & 0.5 &  0.5  \\
\hline
\end{tabular}
\label{tab:attribution}
\end{table}


\addcolor{There are other more sophisticated attribution algorithms. For example, Incremental Value Heuristics (\textbf{IVH}) \cite{singal2019shapley} measures the impact of an impression by comparing results between a group exposed to the ads and a control group that was not exposed. Shapley Value methods (\textbf{SV}) \cite{zhao2018shapley,singal2019shapley} originate from game theory and allocate credit to each impression based on its marginal contribution. Data-Driven Attribution (\textbf{DDA}) \cite{shao2011data,googleDDA} uses machine learning models to analyze historical data and assign attribution weights by comparing counterfactual impression patterns to reveal how different touchpoints work together to drive conversions}. \addcolor{Our Ads-BPC algorithm is compatible with these advanced attribution models, provided that the normalized total weight assigned to each publisher for any given conversion $\conversion$ does not exceed 1.}

\subsubsection{Measuring Attributed Conversions}
\label{subsec:sum}


In order to effectively evaluate the advertisement's performance across multiple publishers, it is crucial to aggregate the total attribution weights, i.e., total conversions, for each publisher. A large cumulative sum of weights for a publisher signifies that advertisements on that platform have a substantial likelihood of resulting in product conversions. This metric becomes instrumental in assessing the effectiveness of ad campaigns and optimizing advertising strategies. 


For a specific publisher and conversion event $\conversion$, the attribution process operates as follows: If $\conversion$ is attributed to an impression $\impression$ from this publisher, the publisher is assigned an attribution weight of $w$. Then we will compute, for all conversions, the total attribution weights received by each publisher. A large sum of attribution weights for a publisher signifies that advertisements on that platform are more likely to be profitable and convert into actual conversions.

\subsubsection{Streaming Measurement}
\label{sec:online}


\begin{table}
\centering
\caption{Streaming Measurement of Attributed Conversions across Publishers via Queries}
\begin{tabular}{|c|c|c|c|c}
\hline
 Pub-ID & Day 1 & Day 2 & Day 3 & $\cdots$ \\
\hline
 P-1 & $\query_1(\vec{x}_{\text{P1}}[1])$& $\query_2(\vec{x}_{\text{P1}}[1:2])$ & $\query_3(\vec{x}_{\text{P1}}[1:3])$ & $\cdots$\\
 P-2 & $\query_1(\vec{x}_{\text{P2}}[1])$& $\query_2(\vec{x}_{\text{P2}}[1:2])$ & $\query_3(\vec{x}_{\text{P2}}[1:3])$ & $\cdots$\\
\hline
\end{tabular}
\label{tab:online}
\end{table}

\addcolor{In practical applications, it is essential to adopt a continuous approach for releasing queries on attribution weights. For instance, an advertiser might want to obtain the cumulative sum of attribution up to the current day and report this sum daily. Alternatively, they might be interested in reporting the sum of attributions for the past week, using a sliding window approach. This approach facilitates the ongoing monitoring of trends over time, allowing for the evaluation of changes in advertising performance.}

Before presenting an example, let's define some key notations. We assume the ad campaign will run for $n$ days, and we represent the data vector as $\vec{x}_{\text{P}} \in \mathbb{R}^{n}$, where $\vec{x}_{\text{P}}[i]$ denotes the sum of attribution weights for publisher P on day $i$ ($1 \leq i \leq n$). For simplicity, we use $\datavec$ to represent the data vector when the publisher P is clear from the context. On day $i$, the advertiser is interested in obtaining an answer to the query $\query_i(\datavec[1:i])$. The function $\query_i$ depends on $\datavec[1:i]$ and has access only to the first $i$ entries of $\datavec$, as the remaining entries are yet to be determined.

\addcolor{Table \ref{tab:online} provides an example of streaming measurement for attributed conversions using online queries. In this scenario, the advertiser Nike is interested in obtaining answers to online queries $\query_1, \query_2, \dots, \query_n$ across two major publishers, Publisher 1 and Publisher 2. This dataset allows Nike to evaluate the performance of their advertisements in real-time, enabling Nike to make timely and informed decisions that can significantly impact their business.}

Our primary concern is to ensure that the released information will be privacy-preserving. Specifically, we require that the released information satisfies differential privacy, a formal privacy guarantee that will be discussed next in Section \ref{sec:dp}. In Section \ref{sec:problemdef}, we will introduce the methodology for achieving this privacy-preserving release of the dataset.

\subsection{Differential Privacy}
\label{sec:dp}

Differential Privacy \cite{dwork06Calibrating,zcdp} establishes a set of constraints governing the operation of a mechanism $\mech$. It has emerged as the de facto standard for safeguarding data privacy when publicly accessible data are released. Differential privacy offers a robust guarantee of plausible deniability by constraining an adversary's ability to ascertain whether a specific individual data record was included in the dataset or not. 


One of the key features of differential privacy is its robust privacy guarantee, which remains effective regardless of an adversary's knowledge or actions when attempting to attack the data. In this context, an adversary refers to any person or entity trying to extract sensitive information about an individual record from the analysis results of the data. This privacy guarantee is maintained even if the adversary possesses unlimited computing power and has full knowledge of the algorithm and system used for data collection and analysis.



A fundamental concept in differential privacy is the notion of a \textit{neighboring dataset}. Let $\mathcal{N}$ represent the set of all possible pairs of neighboring datasets, denoted as $(\datavec, \tilde{\datavec})$. We define the vector $\vec{d}$ as the difference between these neighboring datasets, expressed as $\vec{d} = \datavec - \tilde{\datavec}$. The \textit{diffset} of $\mathcal{N}$ is defined as the collection of all possible $\vec{d}$ values, which we denote as $\diffset = \{\vec{d}: \vec{d} = \datavec - \tilde{\datavec}, \forall (\datavec, \tilde{\datavec}) \in \adj \} $.
The most prevalent variant of differential privacy is
given below.
\begin{definition}[Differential Privacy~\cite{dwork06Calibrating}]\label{def:dp} Given privacy parameters $\epsilon > 0$ and $\delta \in (0, 1)$, and the set of all pairs of neighboring datasets $\adj$, a randomized algorithm $\mech$ satisfies $(\epsilon,\delta)$-differential privacy if for all $(\datavec, \tilde{\datavec}) \in \adj$ and all sets $S$, the following equations hold:
\begin{align*}
    P(\mech(\datavec)\in S) \leq e^{\epsilon} P(\mech(\tilde{\datavec})\in S) + \delta.
\end{align*}
When $\delta=0$, we say $\mech$ satisfies $\epsilon$-differential privacy.
\end{definition}

\noindent We introduce an alternative notion of differential privacy, zero-concentrated differential privacy for analysis.
\begin{definition}[Zero-Concentrated Differential Privacy (zCDP) \cite{zcdp}]\label{def:zcdp} Given a privacy parameter $\rho > 0$ and the set of all pairs of neighboring datasets $\adj$, a mechanism $\mech$ satisfies $\rho$-zCDP if for all $(\datavec, \tilde{\datavec}) \in \adj$ and for all $\alpha \in (1, \infty)$, it holds that $DIV_{\alpha}(\mech(\datavec) || \mech(\tilde{\datavec})) \leq \rho \alpha$, where $DIV_{\alpha}$ is the $\alpha$-Renyi divergence between distributions $\mech(\datavec) $ and $\mech(\tilde{\datavec})$.
\end{definition}

\noindent The following two theorems show the conversion between zCDP and $(\epsilon,\delta)$-DP,  $\epsilon$-DP.
\begin{theorem}[zCDP to DP Conversion \cite{canonne2020discrete}]\label{thm:convert}
If mechanism $\mech$ satisfies $\rho$-zCDP, then it satisfies $(\epsilon, \delta)$-DP for any $\epsilon > 0$ and 
\begin{align*}
    \delta = \min_{\alpha > 1} \frac{\exp ((\alpha - 1)(\alpha \rho - \epsilon)) }{\alpha- 1} (1-\frac{1}{\alpha})^{\alpha}.
\end{align*}
\end{theorem}

\begin{theorem}[Pure DP to zCDP Conversion \cite{DPorg-pdp-to-zcdp}]\label{thm:puredp}
If mechanism $\mech$ satisfies $\epsilon$-DP, then it satisfies $\rho$-zCDP with $\rho = \frac{e^{\epsilon} - 1}{e^{\epsilon} + 1} \epsilon$.
Furthermore, the bound is tight.
\end{theorem}


In this paper, our primary emphasis is on zCDP due to its implementation simplicity. Nevertheless, it's important to highlight that our proposed method seamlessly extends to other variants of differential privacy.
To analyze the privacy of a given mechanism, an important quantity is sensitivity, which measures the maximum change in function value on neighboring datasets. 


\begin{definition}
\label{def:l2sen}
The $L_p$ ($p=1,2$) sensitivity of a function $f: \mathbb{R}^{n} \rightarrow \mathbb{R}^{n}$ is given by $\Delta_p (f)  = \max_{(\datavec, \tilde{\datavec}) \in \adj} \|f(\datavec) - f(\tilde{\datavec}) \|_p$,
where $\| \mat{v}\|_p$ is the $L_p$ norm of the vector $\mat{v}$. 
\end{definition}

\begin{theorem}
    [Gaussian Mechanism of zCDP \cite{zcdp}]\label{thm:gaussian} Let $f: \mathbb{R}^{n} \rightarrow \mathbb{R}^{n}$ be a function of the dataset $\datavec$, the Gaussian mechanism $\mech$ adds i.i.d. Gaussian noise to $f(\datavec)$ with scale parameter $\sigma^2$, $\mech(\datavec) = f(\datavec) +   N(0, \sigma^2 \mat{I}_n),$
    where $\mat{I}_n$ is the $n\times n$ identity matrix. The Gaussian mechanism satisfies $\frac{\Delta_2(f)^2}{2\sigma^2}$-zCDP.
\end{theorem}


A classical approach is to apply this Gaussian mechanism to the vector of measurement results. However, we later demonstrate that such an i.i.d. distribution is not optimal.
\subsubsection{Neighboring Dataset}
\label{sec:neighbor}

We first introduce the event-level adjacency $\subeventadj$. $\datavec$ and $\tilde{\datavec}$ are called event-level neighboring datasets, denoted as $(\datavec, \tilde{\datavec}) \in \subeventadj$, if $\tilde{\datavec}$ can be obtained from $\datavec$ by \textbf{substituting a single record} in $\datavec$. \addcolor{We assume a fixed number of users, and therefore use the replacement notion for defining neighboring datasets. However, with minor adjustments to the sensitivity calculation, our method can also be adapted to the add-or-remove notion.} In this context, the diffset of $\subeventadj$ is defined as $\subeventdiff = \{\vec{d}: \|\vec{d}\|_1 \leq 2\}$. This is because when a single record is substituted, $\tilde{\datavec}$ can deviate from $\datavec$ by 2 entries, with each entry differing by 1, resulting in $\|\tilde{\datavec} - \datavec \|_1 \leq 2$.

However, it is important to recognize that users may not be content with protecting only individual actions taken on the internet. Instead, they often seek comprehensive protection for all of their personal information. Consequently, there arises a need for user-level neighboring datasets.
To move from event-level neighboring datasets to user-level neighboring datasets, we replace the concept of ``the change of one record'' with ``the change of all records associated with one user.'' However, it's important to note that since a user can potentially contribute an infinite number of records, and differential privacy aims to provide worst-case guarantees; introducing excessive noise can significantly impact utility performance.
To address this challenge, a common practice is to constrain user contributions, as discussed in prior works such as \cite{dong2023continual, googlesql, liu2023algorithms}. This constraint limits each user to having at most $k$ records in the dataset. In cases where a user exceeds this limit and contributes more than $k$ records, only the first $k$ records are retained, while the surplus records are discarded.

The user-level neighboring dataset with a global contribution constraint of at most $GS$, denoted as $\subglobaladj$, defines a pair of datasets $(\datavec, \tilde{\datavec})$ as belonging to $\subglobaladj$ if $\tilde{\datavec}$ can be obtained from $\datavec$ by substituting all records (up to a maximum of $GS$ records) of one user from $\datavec$. Because each user can contribute a maximum of $GS$ records, the diffset of $\subglobaladj$ can be characterized as $\subglobaldiff = \{\vec{d}: \|\vec{d}\|_1 \leq 2GS\}$. This is because when $GS$ records are substituted, in the extreme case, $\tilde{\datavec}$ can deviate from $\datavec$ by 2 entries, with each entry differing by $GS$, resulting in $\|\tilde{\datavec} - \datavec \|_1 \leq 2GS$.

\subsubsection{Differential Privacy Properties}

\begin{theorem}[Post-processing \cite{zcdp}] \label{thm:post}
Let $\mech$ be a mechanism, $f$ be an arbitrary function that takes the output of $\mech$ as the input. If the mechanism $\mech$ satisfies $\rho$-zCDP, then the post-processing $f(\mech)$ satisfies $\rho$-zCDP.
\end{theorem}

\begin{theorem}[Sequential composition \cite{zcdp}]
\label{thm:seq}
Let $\mech_1$, $\cdots$, $\mech_k$ be a sequence of mechanisms. If the mechanism $\mech_i$ in the sequence satisfies $\rho_i$-zCDP, then the composition of the sequence $\mech = (\mech_1, \cdots, \mech_k)$ satisfies $(\sum_{i=1}^{k}\rho_i)$-zCDP.
\end{theorem}



\begin{theorem}[Group privacy for zCDP \cite{zcdp}]
\label{thm:group}
Let $\mech$ be a mechanism that satisfies $\rho$-zCDP for any event-level neighboring datasets $(\datavec, \tilde{\datavec}) \in \subeventadj$. Then for any user-level neighboring datasets with a global contribution constraint of at most $k$, $(\datavec, \tilde{\datavec}) \in \subglobaladj$, $\mech$ satisfies $k^2 \rho$-zCDP.
\end{theorem}

Group privacy facilitates the transition from event-level differential privacy to user-level differential privacy. Nevertheless, in the context of online advertising measurement systems, applying group privacy overlooks the likelihood that a user's contributions are more evenly distributed across each day, rather than being concentrated in a single day. By harnessing the patterns of user contributions across days, it becomes feasible to enhance utility. This is the motivation behind our proposal for a user-level differentially private advertising measurement system with bounded per-day contributions, affording greater control on a day-to-day basis.

\section{Related Work}\label{sec:related}

 \noindent
\textbf{A lack of formal privacy guarantee in industry practices.} To address growing consumer concerns over privacy, many advertising networks and platforms have introduced a range of privacy-preserving advertising measurement tools. Notable among these initiatives are Apple's Private Click Measurement (PCM) framework and Google’s Attribution Reporting API, which is part of the Privacy Sandbox. PCM mostly relies on reducing data granularity—for example, by allowing only limited user identifiers or implementing delayed reporting—to decrease the amount of information websites can gather about user activities on other sites. Similarly, the Attribution Reporting API delivers \textit{event-level} reports that share many of PCM's techniques, such as delayed reporting and coarse data granularity. In addition, the Attribution Reporting API also generates summary reports, or aggregated reports, by utilizing encryption and a trusted execution environment (TEE). Neither PCM nor the Attribution Reporting API can provide any \textit{formal privacy guarantee} for the advertising measurement output with a mathematically rigorous definition. It is unclear whether such approaches are indeed useful in safeguarding user privacy since any aggregated report could potentially be subject to privacy risks in a way similar to the US census data attack in a series of works by Dwork and others~\cite{dick2023confidence, keller2023database}. Recently, Google \cite{delaney2024differentially} discussed releasing differentially private conversions in a one-shot setting, whereas our focus is on releasing data in an online setting. Alistair \cite{tholoniat2024alistair} scheduled the on-device privacy budget efficiently to achieve individual differential privacy (IDP), with a special focus on secure multiparty computation (MPC) and TEE.


 \noindent
\textbf{Secure computation is not sufficient.} The advertising industry has recently proposed secure multi-party computation (MPC) based solutions. These solutions ensure the input data is encrypted during advertising measurement computations. Examples include PSI-Sum~\cite{ion2020deploying} by Google and Private Join and Compute~\cite{meta} by Meta, as well as the Interoperable Private Attribution (IPA)~\cite{case2023interoperable}, a collaborative effort by Mozilla and Meta. Despite the encryption of input data, the reports—output from the MPC computation—are plaintext results without formal privacy guarantees. 
Recent studies have shown that in the case of advertising measurement, membership inference attacks can exploit this vulnerability, potentially risking user privacy and business confidentiality~\cite{guo2022birds,Bo}.

\noindent
\textbf{Ensuring privacy guarantees for online reporting poses challenges.} An advertising campaign can extend  even months. Applying differential privacy to the cumulative ad measurement results  ensure the privacy of the final data. However, it is also crucial to report daily measurement results online. Simply applying a differential privacy mechanism by adding independent noise for each day severely degrades measurement accuracy and usability, as demonstrated by the experiments in Section~\ref{sec:experiments}.
One possible research direction to improve the accuracy without sacrificing privacy is to
use tree-based mechanisms to address the online release problem under \textit{event-level} differential privacy. This direction was initially introduced by Chan and Dwork \cite{chan2011private, dwork2010differential}. Subsequent studies by Denisov et al. \cite{denisov2022improved} leveraged the matrix mechanism \cite{LMHMR15,xiao2020optimizing,xiao2023answering} to enhance the accuracy of these tree-based approaches by improving the constant factors in their error rates. \addcolor{However, these methods have limited flexibility, as they primarily focus on minimizing the root mean squared error for prefix-sum queries. In contrast, Ads-BPC can optimize a variety of queries with different objective functions.}

\noindent
\textbf{User-level privacy guarantee is preferred.} Two definitions of differential privacy exist in the advertising setting: event-level and user-level. The former safeguards individual events, while the latter conceals all events associated with a particular user across the entire ads measurement campaign. For instance, event-level privacy protects a single user's click or purchase activity (i.e., conversion), whereas user-level privacy shields all user clicks or conversions. The chosen privacy level impacts the degree of perturbation applied.
Choquette-Choo et al. introduced the matrix mechanism as a solution for user-level differential privacy in their work \cite{choquette2022multi, choquette2024amplified}. Their primary application is differentially private stochastic gradient descent. Nevertheless,these methods necessitate a minimum time gap of at least $b$ days between two records belonging to the same user. This requirement presents a challenge in the context of online advertising measurement systems, where a user can potentially have multiple conversions within a single day. Investigating user-level differential privacy in the context of continual observation, \cite{dong2023continual, zhang2023differentially} used tree-based mechanisms to tackle the problem. Although effective for addressing the general streaming problem, these approaches do not leverage the specific structures inherent in advertising measurement systems. {For example, in a general streaming process, a single user may contribute most of the records, leading to a large sensitivity (cropping the data related to this user will lead to a high bias). However, in an advertising system involving millions of users, most users generate only a small number of conversions in practice . The removal of data associated with a few users is unlikely to have a significant impact on the overall dataset.} In this paper, we focus on the challenge of achieving the more stringent user-level differential privacy while satisfying the constraints and properties of an online ad measurement system.


\section{Differentially Private Measurement Design}\label{sec:problemdef}



We design our system to satisfies the following properties. \textbf{Flexibility}. The goals vary across different types of ad campaigns, and we want to avoid redesigning the system for each specific case. Ideally, the system should accommodate various objectives with only minor modifications. \textbf{Simplicity}. The system should be easy to deploy and should introduce minimal changes to the existing non-private version. \textbf{Adaptability}. The system should automatically adjust to different datasets, ensuring the algorithm performs well across various types of data. Given these requirements, we propose the following mechanism design: 
\begin{enumerate} 
\item Each day $i$, we add random Gaussian noise $z_i$ to obtain the noisy result $\tilde{x}_i = x_i + z_i$, where $z_i \sim N(0, \sigma_i^2)$ is a Gaussian random variable with mean zero and variance $\sigma_i^2$. While simple, adding independent noise—combined with per-day contribution limits—proves more effective in our experiments than more complex methods like the matrix mechanism. 
\item Depending on the specific objective function, we initialize the scale parameters $\bar{\sigma}_i^2$, providing flexibility for different ad campaign goals. 
\item Each day, we privately estimate the user contribution limit $r_i$. A user can contribute up to $r_i$ conversions on day $i$, with any additional conversions removed. 
\item We then adjust the actual noise scale as $\sigma_i = r_i * \bar{\sigma}_i$, based on the contribution limit $r_i$ for day $i$. This data-dependent noise scaling makes the system more adaptive. 
\end{enumerate}


\begin{algorithm}
\caption{Ads Measurement with Bounded Per Day Contributions (Ads-BPC)}
\label{alg:adsbpc}
\begin{algorithmic}[1]
\STATE \textbf{Input:} Target workload query matrix $\mat{Q}$; Objective function OBJ.
\STATE \textbf{Output:} Noisy query answer.
\label{line:scale}
\STATE $\bar{r}_1, \cdots, \bar{r}_n $ $=1$.
\STATE $\bar{\sigma}_1, \cdots, \bar{\sigma}_n$ = InitScales($\mat{Q}$, OBJ, $\{\bar{r}_i \}_{i=1}^{n}$).

\STATE bound\_list = [].

\FOR {$i$ in range($n$)}
    \IF{$i < l$} 
        \STATE $r_i$= PrivateQuantile($D_i$, percent).
    \ELSE
        \STATE $r_i$ = UpdateBoundSVT($D_i$).
    \ENDIF
    \STATE Store the bound: bound\_list.append($r_i$).
    \STATE Clip data: $\datavec[i]$ = Clip($D_i$, $r_i$).
    \STATE Update scale: $\sigma_i= \bar{\sigma}_i / \bar{r}_i *r_i$.
    \label{line:update}
    \STATE Add noise to data: $\tilde{\datavec}[i] = \datavec[i] + N(0, \sigma_i^2)$.
    \STATE Get query answer: $y_j = \query_j^T \tilde{\datavec}$.
\ENDFOR

\STATE Return.
\end{algorithmic}
\end{algorithm}

In the following sections, we describe each of the above components in detail. Algorithm \ref{alg:adsbpc} provides an overview of the \adsmech mechanism. We begin by initializing the scale parameters based on the given query workload and objective function, with Section \ref{sec:init} covering the initialization process in detail. For the first $l$ days, we estimate the user contribution limit $r_i$ using a differentially private quantile algorithm, as explained in Section \ref{sec:quantile}. For the remaining days, we apply the Sparse Vector Technique (SVT) to determine whether an adjustment to the user contribution limit is needed. Section \ref{sec:svt} provides the implementation details for the SVT algorithm. We will begin by addressing the scenario involving a single publisher. In cases where multiple publishers are involved, we show how to extend the mechanism in Section \ref{sec:compose}. Before we introduce details of our algorithm, we show how to measure the utility of the proposed mechanism in Section \ref{sec:utility}, and give privacy properties of the mechanism in Section \ref{sec:budget}

\subsection{Mechanism Utility}
\label{sec:utility}
Recall that there are $n$ days in total, and we represent the data vector as $\vec{x} \in \mathbb{R}^{ n}$, where $\vec{x}[i]$ represents the sum of attribution weights for a specific publisher on day $i$ ($1 \leq i \leq n$). In order to analyze our mechanism using Theorem \ref{thm:gaussian}, we rewrite the noise addition step as $\tilde{\datavec}[i] = \datavec[i] + N(0, \sigma_i^2) = \sigma_i ( \frac{1}{\sigma_i} \datavec[i] + N(0, 1))$. That is, we add a Gaussian random variable with zero mean and variance of one to a scaled version of $\datavec[i]$. Let $\mat{B} = \text{diag} (\frac{1}{\sigma_1}, \cdots, \frac{1}{\sigma_n})$, the private mechanism is defined as,
\begin{align}
    \mech(\datavec) = \mat{B} \datavec + N(0, \mat{I}).
    \label{eq:mechanism}
\end{align}

We represent the $j$-th query vector as $\query_j = (\query_{j}[1], \query_{j}[2], \cdots, \query_{j}[n])$, the true answer to this query is $\query_j^T \datavec$ = $ \query_{j}[1] \datavec[1] + \query_{j}[2] \datavec[2] + \cdots + \query_{j}[n] \datavec[n]$. Because we perturb the ground truth data with independent Gaussian noise, $\tilde{\datavec}[i] = \datavec[i] + N(0, \sigma_i^2)$, the noisy answer to the query $\query_j$ is $\query_j^T \tilde{\datavec} = \sum_{i=1}^{n} \query_j[i] (\datavec[i] + N(0, \sigma^2_i))$ $= \sum_{i=1}^n \query_j[i] \datavec[i] + \sum_{i=1}^{n} \query_j[i] N(0, \sigma^2_i)$ $=\query_j^T \datavec + N(0, \sum_{i=1}^n \query_j[i]^2 \sigma^2_i).$ 
We can see that the variance of the query $\query_j$ is
\begin{align}
\label{mech:var}
    \varfun(\query_j;\mech) := \sum_{i=1}^n \query_j[i]^2 \sigma^2_i.
\end{align}


The variance function of the mechanism $\mech$ is defined as $\loss(\varfun(\query_1;\mech),\dots, \varfun(\query_m;\mech))$, here $\loss$ is a function of $\varfun(\query_j;\mech)$. For example, $\loss$ can be the sum function, then it represents the sum of variances for all queries.

\subsection{Mechanism Privacy Budget}
\label{sec:budget}

As stated in Theorem \ref{thm:gaussian}, the privacy budget is proportional to the $L_2$ sensitivity of the mechanism. Since we are adding independent noise to the data $\datavec_i$ each day, a natural way to reduce sensitivity is to limit the number of contributions a user can make daily. To achieve this, we introduce the concept of bounded per-day contributions.

\begin{definition}[User-Level Neighboring Dataset with Bounded Per-day Contributions]
    \label{def:bpc}
Given the constraint that each user can contribute at most $r_i$ conversions on day $i$ ($1 \leq i \leq n$), the user-level neighboring dataset with bounded per-day contributions, denoted as $\subperdayadj$, defines a pair of datasets $(\datavec, \datavec')$ as belonging to $\subperdayadj$ if $\datavec'$ can be obtained from $\datavec$ by substituting all records (with a maximum of $r_i$ records on day $i$) of one user from $\datavec$. The diffset of $\subperdayadj$ can be expressed as $\subperdaydiff = \{\vec{d} \in \mathbb{R}^{n}: -r_i \leq \vec{d}_i \leq r_i, 1\leq i\leq n\}$.
\end{definition}

\noindent The privacy budget is determined by the sensitivity of the function $f(\datavec) = \mat{B} \datavec$, or more precisely, how much the outputs of this mapping can deviate (in terms of $L_2$ norm) when transitioning from the input data vector $\datavec$ to a neighboring data vector $\datavec'$, under the definition of $\subperdayadj$. The following theorem shows how to calculate the sensitivity of the function $f$.

\begin{theorem}[Sensitivity of Mechanism \ref{eq:mechanism}]
\label{thm:sen}
The \textbf{sensitivity} of the mechanism \ref{eq:mechanism} under the definition of $\subperdayadj$ is
\begin{align}
\nonumber
\Delta_2(f) = & \max_{(\datavec, \datavec') \in \subperdayadj} \|\mat{B} \datavec - \mat{B} \datavec' \|_2 = \max_{\vec{d} \in \subperdaydiff} \| \mat{B} \vec{d}\|_2\\
&= \sqrt{\frac{r_1^2}{\sigma_1^2} + \cdots + \frac{r_n^2}{\sigma_n^2} }.
\end{align}
\end{theorem}

\noindent We define the privacy budget of the mechanism as
\begin{align}
\label{mech:budget}
    \pcost(\mech):= \frac{\Delta^2_2(f)}{2} =\rho. 
\end{align}

\subsection{Initialize the Scales}
\label{sec:init}

With the expressions for utility and budget of our mechanism $\mech$ established in previous sections, we can now optimize the parameters $\sigma_i$ for various queries and objective functions. Before running the ads measurement system, advertisers need to provide: 
1) The specific queries of interest, such as the number of conversions each day or the number of conversions from day 10 to day 20. 
2) The objective function of interest, such as minimizing the root mean squared error or minimizing the maximum variance of the queries.

Depending on the provided queries and objective function, the advertiser may aim to minimize the objective function within the constraints of a given privacy budget. 
\begin{align}
\nonumber
&\arg\min_{\mech} \loss(\varfun(\query_1;\mech),\dots, \varfun(\query_m;\mech)) \\
&\textbf{ s.t. }\quad\pcost(\mech)\leq \rho  \quad \text{(Privacy constrained)}\label{eqn:privacyconstrained}
\end{align}

Alternatively, the advertiser might have accuracy constraints for the queries and seek to minimize the total privacy budget. We use a constant vector $\vec{v}$ to specify variance constraints for queries.
\begin{align}
\nonumber
&\arg\min_{\mech} \pcost(\mech) \quad \quad \text{(Utility constrained)} \\
&\textbf{ s.t. }\quad \loss(\varfun(\query_1;\mech),\dots, \varfun(\query_m;\mech))\leq \vec{v}
\label{eqn:accuracyconstrained}
\end{align}

\noindent Then we can solve for $\sigma_1, \sigma_2, \cdots, \sigma_n$ as the initial solution. The following theorem shows that using initial scales derived above, Mechanism \ref{eq:mechanism} satisfies zCDP.

\begin{theorem}
\label{thm:dp} Let $\mat{Q} =[\query_1;\query_2;\cdots;\query_m] \in \mathbb{R}^{m\times n}$ be the input Query Workload matrix. For Mechanism \ref{eq:mechanism}, let $\rho_1=\pcost(\mech)$, 
$\sigma_1, \sigma_2,$ $ \cdots \sigma_n$ be the noise scales, $ \mat{B} = diag(\frac{1}{\sigma_1}, \frac{1}{\sigma_2} \cdots, \frac{1}{\sigma_n}) $, $\mat{L} = \mat{Q}\mat{B}^{-1}$ and $\mat{z} \sim N(0, \mat{I}_n)$, then $\mat{L} \mech(\datavec)$ = $\mat{L} (\mat{B} \datavec + \mat{z}) = \mat{Q} \datavec+ \mat{L}\mat{z} $ satisfies $\rho_1$-zCDP.
\end{theorem}

We next show an example of calculating the initial scales. Suppose on each day, the advertiser wants to get the total number of conversions up to today, the query $\query_i$ on day $i$ indicates a prefix sum function on $\datavec$, $\query_i[t]=1$, when $1\leq t \leq i$, otherwise $\query_i[t]=0$.


\noindent The advertiser wants to minimize the weighted sum of variances given a fixed privacy budget. The loss function is the weighted sum of variances, $\gamma_i$ is the weight for $\query_i$.

\begin{align*}
    &\loss(\varfun(\query_1;\mech),\dots, \varfun(\query_n;\mech)) 
    = \sum_{i=1}^n \gamma_i^2 \varfun(\query_i;\mech)  \\
    =&\sum_{i=1}^n \left(\gamma_i^2 \sum_{t=1}^{i} \sigma_t^2 \right)
    =\sum_{i=1}^n  \left(\sum_{j=i}^{n} \gamma_j^2 \right) \sigma_i^2.
\end{align*} 


\noindent The constraint on privacy budget is $\pcost(\mech) \leq \rho$, or equivalently, $\Delta_2^2(f)\leq 2\rho$. 
Because we want to achieve $\rho$-zCDP while minimizing the total variance, the problem becomes,
\begin{align}
\label{eq:cauchy}
\nonumber
    \min_{\sigma_1^2, \cdots, \sigma_n^2} \quad & \left(\sum_{j=1}^{n} \gamma_j^2 \right) \sigma_1^2 + \left(\sum_{j=2}^{n} \gamma_j^2 \right)\sigma_2^2 + \cdots + \gamma_n^2\sigma_n^2 \\
    s.t. \quad &\frac{r_1^2}{\sigma_1^2} + \frac{r_2^2}{\sigma_2^2} + \cdots + \frac{r_n^2}{\sigma_n^2}  \leq 2\rho. 
\end{align}
Problem (\ref{eq:cauchy}) can be solved analytically using Cauchy inequality \cite{xiao2024optimal}. 

We give another example in the form of Problem \ref{eqn:accuracyconstrained}. In this example, the advertiser has requirements for the accuracy of each query $\query_j$, the goal is to find a mechanism that has minimal privacy cost.

\begin{align}
\label{eq:max}
    \min_{\sigma_1^2, \cdots, \sigma_n^2} \quad & \frac{r_1^2}{\sigma_1^2} + \frac{r_2^2}{\sigma_2^2} + \cdots + \frac{r_n^2}{\sigma_n^2}  \\
    \nonumber
    s.t. \quad & \varfun(\query_j;\mech) \leq v_j \quad j = 1, \cdots, m. 
\end{align}
As shown in Equation \ref{mech:var}, the variance function is a linear combination of $\sigma_i^2$. Consequently, Problem \ref{eq:max} is a convex optimization problem that can be efficiently solved using any standard solver.

By selecting the initial scales, we maintain flexibility for different types of queries and objective functions. The key decision lies in determining these initial scales. We initialize the user contribution bound as $\bar{r}_i=1$ to compute the initial noise scale $\bar{\sigma}_i$. Once the actual user contribution bound $r_i$ is determined, we update the noise scales using $\sigma_i = \bar{\sigma}_i / \bar{r}_i * r_i$, as shown in Line \ref{line:update} of Algorithm \ref{alg:adsbpc}. This adjustment ensures that the privacy budget of the mechanism $\pcost(\mech)$ remains unchanged. In the following sections, we describe how to adjust the noise scales based on the actual dataset to achieve more accurate measurements.

\subsection{Estimate User Contribution Bound}
\label{sec:quantile}
Selecting an appropriate user contribution bound $r$ involves a trade-off between bias and variance. A small $r$ leads to the removal of numerous records, introducing substantial bias. Conversely, a large $r$ retains most records in the dataset but introduces significant noise into the results, resulting in high variance. To achieve a balance between bias and variance, we employ the differentially private quantile algorithm \cite{gillenwater2021differentially, smith2011privacy} to estimate the contribution bound. Each day, we count the number of contributions per individual and select an appropriate private quantile as the user contribution bound for that day. The differentially private quantile is based on the exponential mechanism \cite{exponentialMechanism}.

\begin{theorem}[Exponential Mechanism \cite{exponentialMechanism,DPorg-exponential-mechanism-bounded-range}] \label{thm:exp} Given utility function $u$ mapping (dataset, output) pairs to real valued scores with $L_1$ sensitivity $\Delta_1(u) = \max_{X \sim X'} |u(X, o) - u(X', o)|$, the Exponential Mechanism $M$ has output distribution $ \mathcal{P}_{M} [M(X) = o] \propto \exp \left( \frac{\epsilon u(X, o)}{2 \Delta_1(u)} \right)$. Here $\propto$ ignores the normalization factor. Then the mechanism $M$ is $\epsilon$-DP, and it satisfies $\frac{1}{8} \epsilon^2$-zCDP.
\end{theorem}


\noindent The utility function for the quantile is defined as 
\begin{align}
    u(X, o) = -||\{ x \in X | x\leq o \}| - p * k|.
\end{align}
For a given output $o$, $|\{ x \in X | x\leq o \}|$ is the count of records with values less than or equal to $o$. Here, $p$ denotes the quantile percentage, and $k$ is the total number of records. As demonstrated in \cite{gillenwater2021differentially, smith2011privacy}, the sensitivity of the utility function described above is 1. Algorithm \ref{alg:quantile} gives a detailed description of of PrivateQuantile. We can prove that the algorithm satisfies zCDP.






\begin{theorem}\label{thm:quantile}
Algorithm PrivateQuantile \ref{alg:quantile} satisfies $\rho_2$-zCDP with privacy paramter $\rho_2 = \min(\frac{1}{8} \epsilon^2, \frac{e^{\epsilon} - 1}{e^{\epsilon} + 1} \epsilon)$.
\end{theorem}


\subsection{Track the Change of Contribution Bound}
\label{sec:svt}

If we run the private quantile algorithm every day, the privacy cost will be significant. We observe that the user contribution bound is stable most of the time, which means that the everyday estimate for the quantile is not necessary. If we keep the bound as a constant and only update it when there's a significant change, it's possible to save a large amount of privacy budget. In case when the contribution bound changes a lot in one day, we use Sparse Vector Technique (SVT) \cite{lyu2017understanding,zhu2020improving,ding2023free} to track such a change. We only need to report the user contribution bound when such a change is detected.

\begin{algorithm}
\caption{CheckUpdate}
\begin{algorithmic}[1]
\label{alg:check}
\STATE \textbf{Input:} Query $q_i$; Privacy budget $\epsilon$; Number of reports $\text{count}$; Maximum number of reports $k$; Noisy threshold $\tilde{T}$; True threshold $T$.
\STATE \textbf{Output:} 

\STATE update = False. 
\IF{$\text{count} < k$} 
   \STATE  \COMMENT{Stop if maximum reports reached.}
\STATE Get noisy query: $\tilde{q}_i = q_i + \text{Lap}(4k/\epsilon)$.
\IF{$\tilde{q}_i > \tilde{T}$}
    \STATE update = True.
    \STATE $\text{count} = \text{count}  + 1$.
    \STATE $\tilde{T} = T+ \text{Lap}(2/\epsilon)$.
\ENDIF
\ENDIF

\STATE \textbf{Return} (update, count, $\widetilde{T}$)
\end{algorithmic}
\end{algorithm}

\begin{algorithm}
\caption{UpdateBoundSVT}
\begin{algorithmic}[1]
\label{alg:svt}
\STATE \textbf{Input:}  Dataset $D_i$.
\STATE \textbf{Parameter:} Privacy budget $\epsilon$; Maximum number of reports $k$; Threshold $T^{\uparrow}, T^{\downarrow}$; Scales $\gap^{\uparrow}, \gap^{\downarrow}$
\STATE \textbf{Output:} The user contribution bound $r_i$.
\STATE \textbf{Global Variable:} Noisy threshold $\widetilde{T}^{\uparrow}, \widetilde{T}^{\downarrow}$;
The number of reports $\text{count}^{\uparrow}, \text{count}^{\downarrow}$; Previous bounds: bound\_list;

\STATE Split $\epsilon$ into $\epsilon^{\uparrow} = \epsilon^{\downarrow}= \frac{\epsilon}{2}$.

\STATE $\tau_i$ = mean(bound\_list[-$l$:]).
\STATE $q^{\uparrow}_i = \abovet(D_i, \tau_i) $.
\STATE (is\_up, $\text{count}^{\uparrow}$, $\widetilde{T}^{\uparrow}$) \\= CheckUpdate($q^{\uparrow}_i$, $\epsilon^{\uparrow}$, $\text{count}^{\uparrow}$, $k$, $\widetilde{T}^{\uparrow}$, $T^{\uparrow}$).

\STATE $q^{\downarrow}_i =  \abovet(D_i, \tau_i) - \abovet(D_i, \tau_i * \gap^{\downarrow})$.
\STATE (is\_down, $\text{count}^{\downarrow}$, $-\widetilde{T}^{\downarrow}$) \\= CheckUpdate($q^{\downarrow}_i$, $\epsilon^{\downarrow}$, $\text{count}^{\downarrow}$, $k$, $-\widetilde{T}^{\downarrow}$, $-T^{\downarrow}$).

\IF{is\_up == is\_down}
\STATE $r_i = \tau_i$. 
\ELSIF{is\_up} \STATE $r_i = \tau_i * \gap^{\uparrow}$.
\ELSIF{is\_down} \STATE $r_i = \tau_i * \gap^{\downarrow}$.
\ENDIF

\STATE bound\_list.append($r_i$).
\STATE Return $r_i$.
\end{algorithmic}
\end{algorithm}

Define $\abovet(D_i, \tau_i)$ as the number of people who has more than $\tau_i$ contributions on day $i$. We can run two parallel SVT to check if there's a need for updating the contribution bound. One SVT checks if it's necessary to increase the bound, while the other one checks if it's better to decrease the bound. Algorithm \ref{alg:check} and Algorithm \ref{alg:svt} outline the details of the SVT. The three primary components are discussed below.

\noindent \textbf{Check if the Bound Increases} In the first SVT, we use the query function $q^{\uparrow}_i = \abovet(D_i, \tau_i) $, here $\tau_i$ is the default contribution bound. We want to check if $\tau_i * \gap^{\uparrow}$ is a better contribution bound, here $\gap^{\uparrow} > 1$ is a scaling parameter. If $q^{\uparrow}_i$ is larger than a predefined threshold $T^{\uparrow}$, it means that updating $\tau * \gap^{\uparrow}$ as the contribution bound can reduce the bias. Therefore we use $\tau_i * s^{\uparrow}$ as the new contribution bound on day $i$. Otherwise, we use $\tau_i$ as the contribution bound.

\noindent \textbf{Check if the Bound Decreases} In the second SVT, we use the query function $q^{\downarrow}_i =  \abovet(D_i, \tau_i) - \abovet(D_i, \tau_i * \gap^{\downarrow})$, here $\tau_i$ is the default contribution bound, and $\gap^{\downarrow} < 1$. We want to check if $\tau_i * \gap^{\downarrow}$ is a better contribution bound. If $q_i$ is greater than the threshold $-T^{\downarrow}$ (or equivalent, $-q_i \leq T^{\downarrow}$), it means that updating $\tau * \gap^{\downarrow}$ as the contribution bound can reduce the variance, without introducing too much bias. Therefore we use $\tau_i * s^{\downarrow}$ as the new contribution bound on day $i$. Otherwise, we use $\tau_i$ as the contribution bound.

\noindent \textbf{Update the Bound} For the first $l$ days, we use differentially private quantile algorithm to calculate the private contribution bound as $\tau_i$. On $i > l$ we update the default bound as the average of previous $l$ bounds $\tau_i = \frac{1}{l}\sum_{j=i-l}^{i-1} \tau_{j}$. If both SVTs attempt to update the bound simultaneously, with one trying to increase it and the other trying to decrease it, we consider this an invalid attempt. In that case, we will use the default bound. 

The detailed description of UpdatedBoundSVT is shown in Algorithm \ref{alg:svt}, we can prove that it satisfies zCDP. 

\begin{theorem}
    \label{thm:svt}
    Algorithm UpdateBoundSVT \ref{alg:svt} satisfies $\rho_3$-zCDP with $\rho_3 = \frac{e^{\epsilon} - 1}{e^{\epsilon} + 1} \epsilon$.
\end{theorem}

\begin{table}
\centering
\begin{tabular}{|c|c|c|c|c|c|c|c|}
    \hline
    $\tau$ & 1 & 2 & 5 & 8 & 10 & 13 & 20 \\
    \hline
    Hist[$D, \tau$] & 1512 & 208 & 53 & 50 & 25 & 12 & 3 \\ 
    \hline
    $\abovet(D, \tau)$ & 351 & 143 & 90 & 40 & 15 & 3 & 0\\ 
    \hline  
\end{tabular}
\caption{Statistics of an example dataset $D$. $\tau$ is the contribution bound. Hist$[D, \tau]$ counts the number of people with $\tau$ contributions. $\abovet(D, \tau)$ counts the number of people with more than $\tau$ contributions}
\label{tab:svt}
\end{table}

\noindent \textbf{An Example for SVT} To illustrate how SVT can detect changes in user contribution bounds, we use an example dataset $D$. For simplicity and clarity, noise is not added to the queries or thresholds in this demonstration. Table \ref{tab:svt} shows the statistics of the example dataset $D$. Suppose the default contribution bound is $\tau = 10$, the gap parameter is $\gap^{\uparrow} = 1.5$ and the threshold for SVT is $T^{\uparrow} = 100$. We want to check if $\tau * \gap^{\uparrow} = 15$ is a better bound. The query value is $q^{\uparrow} = \abovet(D, \tau) = \abovet(D, 10)  = 15$. Because the number of people with contributions above 10 is small (15 is smaller than the threshold $T^{\uparrow}=100$), increasing the default contribution bound can't reduce enough bias, we don't want to increase the bound. 

Now let's check if decreasing the default bound is better. The default contribution bound is $\tau = 10$, suppose the gap parameter is $\gap^{\downarrow} = 0.5$ and the threshold for SVT is $T^{\downarrow} = 100$. We want to check if $\tau * \gap^{\downarrow} = 5$ is a better bound. The query value is $q^{\downarrow} = \abovet(D, \tau) - \abovet(D, \tau *\gap^{\downarrow}) = \abovet(D, 10) - \abovet(D, 5) = -75$. Because the number of people with contributions between $[5, 10]$ is small (75 is smaller than the threshold $T^{\downarrow}=100$), decreasing the default contribution bound can reduce enough variance, we can decrease the bound. 

\subsection{Extension for Multiple Publishers}\label{sec:compose}
The previous sections assume that only one publisher is involved. In practice, multiple publishers often run the ad campaign. In this section, we demonstrate that it is straightforward to extend the previous algorithm to accommodate the multi-publisher case.

Suppose there are $k$ publishers, the data vector for publisher $i$ is $\datavec_i \in \mathbb{R}^{n\times1}$. Let $\mat{X} \in \mathbb{R}^{n\times k}$ be the concatenation of these $k$ data vectors $\datavec_i$. Similar to Section \ref{sec:init}, the goal is to find a private answer to $\mat{Q} \mat{X}$. We factorize the workload matrix $\mat{Q} = \mat{L} \mat{B}$, where $\mat{B} = diag(\frac{1}{\sigma_1}, \frac{1}{\sigma_2} \cdots, \frac{1}{\sigma_n})$ and $\mat{L} = \mat{Q} \mat{B}^{-1}$. Let $\mat{Z} \sim MVG_{n, k}(0, \mat{I}_{ n}, \mat{I}_k)$, where $MVG_{n, k}(0, \mat{I}_{ n}, \mat{I}_k)$ is a Matrix Gaussian Distribution \cite{chanyaswad2018mvg, yang2021improved} (which generates $N(0, 1)$ noise in each entry of $\mat{Z}$), then the mechanism is 

\begin{align}
\label{eq:mat-pub}
    \mat{L}\mech(\mat{X}) = \mat{L} (\mat{B} \mat{X} + \mat{Z}) = \mat{Q} \mat{X}+ \mat{L}\mat{Z}.
\end{align}

Given the constraint that each user can contribute at most $r_i$ records on day $i$ ($1 \leq i \leq n$), the user-level neighboring dataset with bounded per-day contributions, denoted as $\subperdayadj$, defines a pair of datasets $(\mat{X}, \tilde{\mat{X}})$ as belonging to $\subperdayadj$ if $\tilde{\mat{X}}$ can be obtained from $\mat{X}$ by substituting all records (with a maximum of $r_i$ records on day $i$) of one user from $\mat{X}$. The diffset of $\subperdayadj$ can be expressed as $\subperdaydiff = \{\vec{D} \in \mathbb{R}^{n\times k}: 0 \leq \|\vec{D}[i, :]\|_1 \leq r_i, 1\leq i\leq n\}$.

The sensitivity calculation is similar, given that the user contribution for each day is at most $r_i$ (across all publishers), 
\begin{theorem}[Sensitivity of Mechanism \ref{eq:mat-pub}]
\label{thm:sen-pub}
The \textbf{sensitivity} of the mechanism \ref{eq:mat-pub} under the definition of $\subperdayadj$ is
\begin{align}
\Delta_2(f) = \max_{(\mat{X}, \tilde{\mat{X}}) \in \subperdayadj} \|\mat{B} \mat{X} - \mat{B} \tilde{\mat{X}} \|_F 
= \sqrt{\frac{r_1^2}{\sigma_1^2} + \cdots + \frac{r_n^2}{\sigma_n^2} }.
\end{align}
\end{theorem}

Then similar to Theorem \ref{thm:dp}, we can show that the mechanism \ref{eq:mat-pub} also satisfies zCDP,

\begin{theorem}
\label{thm:dp-pub}
Let $\mat{Q} \in \mathbb{R}^{m\times n}$ be the input Query Workload matrix. For Mechanism defined in \ref{eq:mat-pub}, let $\rho_1=\pcost(\mech)$, 
$\sigma_1, \sigma_2,$ $ \cdots \sigma_n$ be the noise scales, $ \mat{B} = diag(\frac{1}{\sigma_1}, \frac{1}{\sigma_2} \cdots, \frac{1}{\sigma_n}) $, $\mat{L} = \mat{Q}\mat{B}^{-1}$ and $\mat{Z} \sim MVG_{n, k}(0, \mat{I}_{n}, \mat{I}_k)$, then $\mat{L}\mech(\mat{X})$ = $\mat{L} (\mat{B} \mat{X} + \mat{Z}) = \mat{Q} \mat{X}+ \mat{L}\mat{Z} $ satisfies $\rho_1$-zCDP.
\end{theorem}

\subsection{Combining Together}

For Algorithm \ref{alg:quantile} and Algorithm $\ref{alg:svt}$, we convert them from $\epsilon$-DP to $\rho$-zCDP using Theorem \ref{thm:puredp}. We prove in the Theorem \ref{thm:adsbpc} that \adsmech satisfies $\rho$-zCDP. 

\begin{theorem}
\label{thm:adsbpc}
    Algorithm AdsBPC \ref{alg:adsbpc} satisfies $\rho$-zCDP, with parameter $\rho = \rho_1 + l*\rho_2 + \rho_3$.
\end{theorem}


\section{Discussion}\label{sec:discuss}

In AdsBPC, we add independent noise to the data vector with different noise scales. A natural question arises: why not use the matrix mechanism to optimize these noise scales? Our reasons are three-fold. \textbf{Complexity}: Solving the exact matrix mechanism under user-level differential privacy is NP-hard \cite{choquette2022multi}. We compare AdsBPC with two variants of the matrix mechanism that approximate the user-level constraint in our experiments. AdsBPC consistently performs better, while also being simpler—eliminating the need for matrix optimization. \textbf{Optimality for Small $n$}: When the number of days $n$ is small and the objective function is the root mean squared error, we calculated the optimal matrix factorization using the code in \cite{choquette2022multi}. Surprisingly, the resulting strategy matrix is diagonal, yielding the same objective value as AdsBPC. See the Appendix \ref{app:opt} for detailed numerical results. Whether AdsBPC achieves optimal noise scales in general remains a question for future research. \textbf{Flexibility}: The matrix mechanism is designed to minimize the sum of variances, limiting it to a single objective function. AdsBPC, on the other hand, offers greater flexibility to tailor to various objectives.


{

Next, we discuss the differences in DP-quantile usage between DP-SGD \cite{abadi2016deep} and advertising systems. While DP-quantile has been applied to ML training \cite{andrew2021differentially}, the advertising scenario presents distinct challenges. In ML training, DP-quantile primarily determines the gradient norm, and \cite{andrew2021differentially} shows that even a minimal privacy budget can yield benefits. However, in advertising, conversion measurement accuracy is highly sensitive to the user contribution bound. A low privacy budget for estimating this bound introduces significant noise. If the bound is too small, excessive data cutoff causes high bias; if too large, increased noise leads to higher variance. Thus, allocating a sufficient privacy budget for DP-quantile is crucial. However, daily application depletes the budget. To mitigate this, we exploit the stability of user contribution quantiles, updating them only when significant changes occur. Our novel SVT-based approach tracks users exceeding a threshold, reducing budget consumption compared to daily DP-quantile updates.

Lastly, we distinguish prior work on DP-SGD using matrix mechanisms \cite{choquette2024amplified,choquette2022multi,denisov2022improved}, emphasizing the fundamental differences between ML training and ads measurement, particularly in how user-level privacy and neighboring datasets are defined. In \cite{choquette2024amplified}, a neighboring dataset is defined such that two samples from the same user must be separated by at least b steps. This constraint is valid in ML training, where sample order can be controlled. However, in ads measurement, this constraint does not hold, as users may have consecutive conversions, and the order of their data cannot be controlled. As noted in \cite{choquette2022multi}, solving a user-level DP problem without specific ordering constraints—such as in our case—is NP-hard. Both \cite{choquette2024amplified} and \cite{choquette2022multi} introduce ordering constraints to make the matrix mechanism feasible. It is also important to note that \cite{denisov2022improved} does not implement a user-level DP mechanism. The key intuition behind our approach comes from the unique characteristics of ads measurement. Unlike ML training, where updates depend on previous iterations, user contributions on different days can be independent—i.e., for a random user, the probability of contributing each day is equal. Furthermore, the sensitivity in our setting differs from that in ML training due to the distinct definitions of neighboring datasets. In our setting, a neighboring dataset is defined by replacing a user's conversions, with a maximum of r conversions allowed per day. This difference in sensitivity makes our approach specifically tailored to ads measurement.

}

\section{Experiments}\label{sec:experiments}

Throughout the experiments, we denote the length of the data vector as $n = len(\datavec)$. We repeat our experiments 10 times and report the average results. All experiments are conducted on a desktop with an 8-core 16-thread 3.6GHz Intel i9-9900K CPU and 32 GB memory.

\subsection{Query and Objective Function}
\label{sec:scene}
We would like to test the performance of each method under two scenario.

\noindent \textbf{Scenario 1.} The advertiser aims to track the trend of conversions over time. Each day, they seek a privacy-preserving estimate of the cumulative conversions up to that day, which corresponds to a prefix-sum query. Additionally, they prioritize accuracy in the final day’s estimate, so the error metric used is the weighted root mean square error (WRMSE), with a higher weight assigned to the last day's query. In this setting, the initialization of AdsBPC follows Problem \ref{eq:cauchy}, we set $\gamma_1 = \cdots =\gamma_{n-1} = 1$ and $\gamma_n = 7$. 

\noindent \textbf{Scenario 2.} The advertiser is interested in the total number of conversions over the past $K$ days, represented by a sliding window query. Each query is treated as equally important, and the error metric used is the maximum variance. To initialize AdsBPC in this context, we follow Problem \ref{eq:max} and set $v_1 = \cdots = v_{n} = 1$. For comparison with other methods, we scale the noise parameters to ensure that the total privacy budget $\rho$ is consistent across all methods. We then compare the maximum variances across all queries for each method.

\subsection{Baseline Methods}

We compare our algorithm (named \textbf{\matmech}) with other baselines. Unless otherwise specified, we allocate a total privacy budget of $\rho_{total} = 1$. Some baseline methods require bounding the global sensitivity. For each dataset, we assume the global sensitivity is given as $GS$. If the global sensitivity is unknown, we can set a reasonably large $GS$ as a threshold. For any user with more than $GS$ records in the dataset, only the first $GS$ records are retained, and any additional records are removed. It is important to note that our algorithm \matmech does not require knowledge of $GS$.

\noindent\textbf{\idglobal} \cite{case2023interoperable}: 
The Interoperable Private Attribution (IPA) protocol leverages a combination of multi-party computation and differential privacy to facilitate the processing of individuals' data while ensuring that only aggregated measurements are disclosed. In this paper, our primary focus is on comparing how IPA achieves differential privacy, assuming the presence of a {trusted third party} that provides the joined table. IPA adds independent and identically distributed (IID) noise to the data vector $\datavec$, we use the group privacy property \ref{thm:group} to make sure it satisfies user-level neighboring datasets with a global contribution constraint of at most $GS$.

\noindent\textbf{\bintree}: \bintree adds IID noise to a binary tree structure, we use the group privacy property \ref{thm:group} to make sure it satisfies user-level neighboring datasets with a global contribution constraint of at most $GS$.

\noindent\textbf{\bincrop} \cite{dong2023continual}: Similar to \bintree, \bincrop adds IID noise to a binary tree structure. Rather than using a noise scale proportional to the global sensitivity $GS$, \bincrop employs the Sparse Vector Technique (SVT) to privately estimate the maximum number of contributions $\tau$ that a user can make in the dataset. Specifically, if a user exceeds $\tau$ contributions, \bincrop doubles the bound, updating $\tau \gets 2 \times \tau$.

\noindent\textbf{\matglobal} (user-level matrix mechanism): It's adapted from the event-level matrix mechanism \cite{henzinger2022constant}, we use the group privacy property \ref{thm:group} to make sure it satisfies user-level neighboring datasets with a global contribution constraint of at most $GS$.

\noindent\textbf{\mmbpc} (matrix mechanism with bounded per-day contribution): Similar to \adsmech, we set a maximum contribution limit $r_i$ for each user on day $i$. If any user contributes more than $r_i$ conversions on a given day, only the first $r_i$ conversions are retained in the dataset. We first compute the event-level matrix mechanism and then scale the columns of the matrix based on the contribution limit $r_i$, ensuring that the sensitivity remains unchanged.

\subsection{Datasets}
\label{sec:dataset}
We use three synthetic datasets and two real world datasets to test the performance of our algorithm. For synthetic datasets, we assume last-touch attribution (LTA) is used. For real datasets, last-touch attribution (LTA), first-touch attribution (FTA) and uniform atribution (UNI) are used to create different data vectors $\datavec$. Table \ref{tab:synthetic-dataset} shows an example of the dataset. At each time (we assume each conversion happens in a non-overlapped timestamp), there is one conversion, that's attributed to one of the publishers.

\begin{table}
    \centering
    \begin{tabular}{|c |c| c| c| c|c}
    \hline
         Time&  1 & 2 & 3 & 4 &  $\cdots$\\
         \hline
        User&  Alice & Bob & Alice & Bob & $\cdots$\\
        \hline
        Publisher 1 & 0 & 1 & 0 & 0 & $\cdots$\\
        \hline
        Publisher 2  & 1 & 0 & 0 & 0 & $\cdots$\\
        \hline
        $\cdots$  & $\cdots$  & $\cdots$ & $\cdots$ & $\cdots$  & $\cdots$\\
        \hline
        Publisher 1000  & 0  & 0 & 0 & 1  & $\cdots$\\
        \hline
    \end{tabular}
    \caption{Attribution weight for each publisher}
    \label{tab:synthetic-dataset}
\end{table}

\begin{table*}
\centering
\begin{tabular}{|c|r|r|r|r|r|r|r|}
\hline
Dataset & IPA & \bintree & \bincrop & UMM & MMBPC &AdsBPC & Improvement \\
\hline
Zipf     & 42.54 & 43.79 & 129.69 & 37.89 & 45.63 & \textbf{21.09} & 44.33\% \\  
Normal      & 189.73 & 142.53 & 427.82 & 123.41 & 116.12 & \textbf{77.28} & 33.45\% \\ 
Uniform     & 294.17 & 153.13 & 760.19 & 316.65 & 161.71 & \textbf{56.69}& 62.98\%\\ 
\hline
Criteo (LTA)    & 42.86 & 36.70 & 68.97 & 38.29 & 30.84 & \textbf{18.23} & 40.88\% \\ 
Criteo (FTA)    & 81.14 & 37.47 & 68.02 & 43.08 & 50.01 & \textbf{21.78} & 41.87\% \\  
Criteo (UNI)    & 58.27 & 35.27 & 63.98 & 34.92 & 41.85 & \textbf{17.02} & 51.26\% \\ 
Facebook     & 46.93 & 34.30 & 74.25 & 37.52 & 20.99 & \textbf{9.78} & 53.41\% \\ 
\hline
\end{tabular}
\caption{The query is the prefix sum, the objective (error metric) is the weighted root mean squared error. The total privacy budget is $\rho=1$. The number of days is $n=31$.}
\label{exp:prefix}
\end{table*}

\begin{table*}
\centering
\begin{tabular}{|c|r|r|r|r|r|r|r|}
\hline
Dataset & IPA & \bintree & \bincrop & UMM & MMBPC &AdsBPC & Improvement \\
\hline
Zipf      & 0.89 & 1.13 & 31.41 & 0.72 & 6.62 & \textbf{0.24} & 66.67\%\\ 
Normal      & 29.72 & 19.19 & 266.40 & 16.65 & 111.83 & \textbf{1.35} & 91.89\% \\ 
Uniform      & 32.66 & 171.18 & 1634.14 & 22.00 & 533.23 & \textbf{1.16} & 94.72\% \\ 
\hline
Criteo (LTA)          & 0.90 & 2.86 & 15.01 & 1.44 & 23.95 & \textbf{0.13} & 85.56\% \\  
Criteo (FTA)        & 1.41 & 2.12 & 13.95 & 0.76 & 25.39 & \textbf{0.11} &  85.53\% \\ 
Criteo (UNI)      & 1.00 & 2.32 & 7.02 & 0.69 & 7.52 & \textbf{0.12} &  82.61\% \\ 
Facebook      & 0.89 & 4.40 & 10.86 & 1.12 & 0.26 & \textbf{0.06} & 76.92\% \\ 
\hline
\end{tabular}
\caption{The query is the sliding window sum, the objective (error metric) is the max variance (divided by $10^4$). The total privacy budget is $\rho=1$. The number of days is $n=31$.}
\label{exp:sliding}
\end{table*}

\begin{figure*}[]
\centering
    \includegraphics[width=\linewidth]{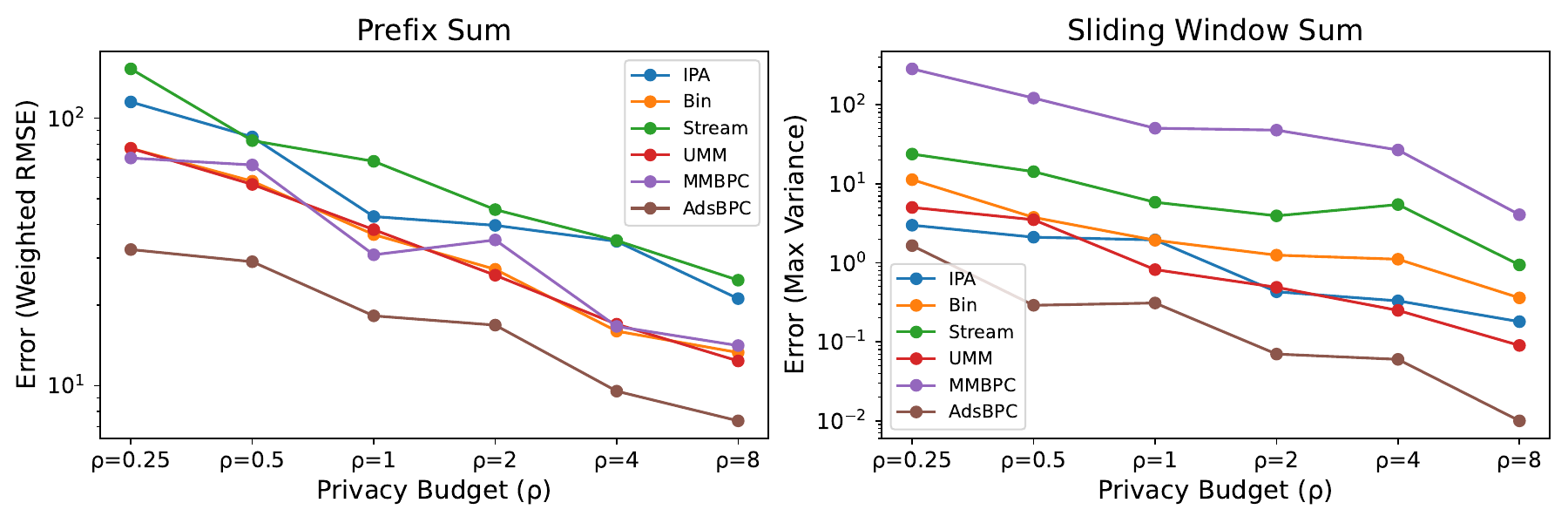}
  \caption{Comparison of different privacy budgets $\rho$ on Criteo Dataset. The number of days is $n=31$.}
  \label{fig:criteo-rho}
\end{figure*}

\begin{figure*}[]
\centering
    \includegraphics[width=\linewidth]{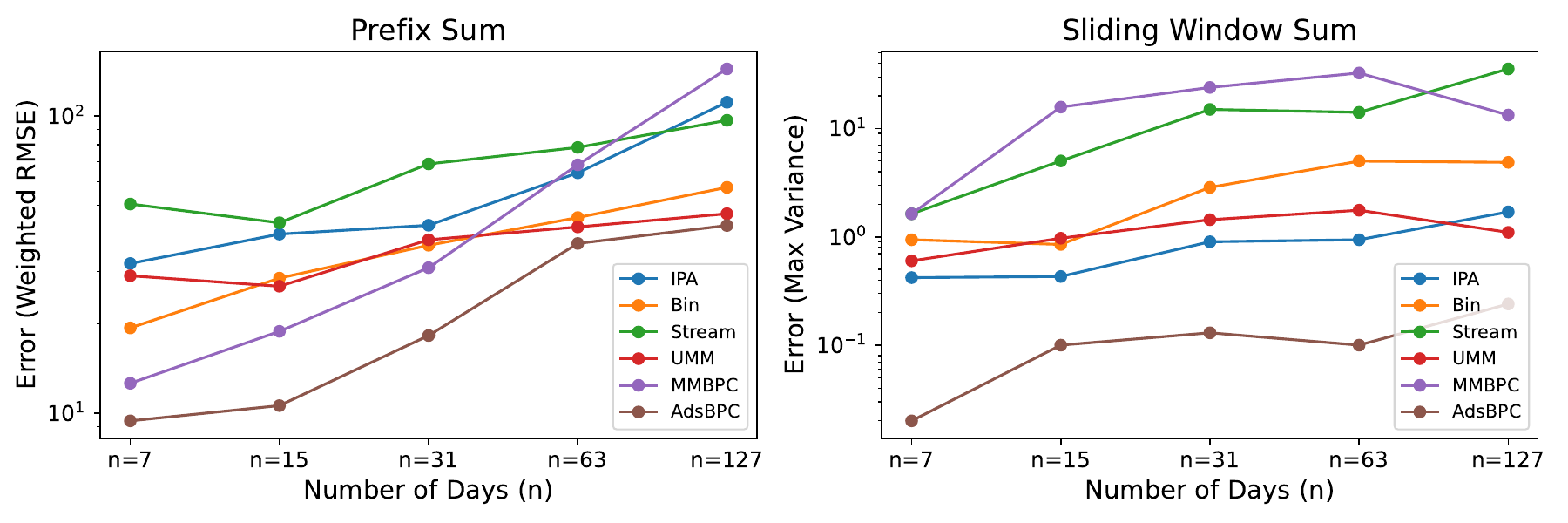}
  \caption{Comparison of different number of days $n$ on Criteo Dataset. The total privacy budget is $\rho=1$. }
  \label{fig:criteo-day}
\end{figure*}

\noindent \textbf{Synthetic Datasets} For synthetic datasets, we assume there are 1 million users, $|U| = 10^{6}$. There are 1000 publishers. For each user, we use 3 different distributions to generate the number of total conversions in the dataset (the distributions are similar to the settings in \cite{dong2023continual}). For example, if the distribution generates a number 5 for the user Alice, it means that Alice has 5 conversions in the dataset. We use the following three syntehtic datasets in the experiment. \underline{\textit{Zipf Dataset}} In the Zipf dataset, the number of conversions for each user is generated from the distribution: Zipf(3)+10. Here the global sensitivity is GS = 50 (A user has at most 50 conversions for each publisher). The total number of conversions in the dataset is $T=11,361,742$. \underline{\textit{Normal Dataset}} The number of conversions for each user is generated from the distribution: Normal(50, 30). Here the global sensitivity is GS = 150. The total number of conversions in the dataset is $T=50,193,976$. \underline{\textit{Uniform Dataset}} The number of conversions for each user is generated from the distribution: Uniform(1, 256). Here the global sensitivity is $GS = 256$, The total number of conversions in the dataset is $T=51,180,247$.

\noindent \textbf{Criteo Dataset}
The Criteo Sponsored Search Conversion Log Dataset \cite{tallis2018reacting} contains logs obtained from Criteo Predictive Search (CPS). Each row in the dataset represents a user action, such as a click, performed on a product-related advertisement, it includes information about product characteristics (age, brand, gender, price), the timing of the click, user characteristics, and device information.

\noindent For simplicity, we assume that the dataset pertains to a single product. The columns of interest are as follows: $\langle \text{sale} \rangle$: This column takes the value 1 if a conversion occurred and 0 if it did not. We aim to attribute the conversion to a specific publisher. $\langle \text{partner\_id} \rangle$: This is a unique identifier associated with the publisher. There are a total of 287 publishers. $\langle \text{user\_id} \rangle$: This unique identifier is associated with each user and is used to track the number of user contributions in the dataset. Key dataset statistics include $\text{GS} = 44$, $T = 1,732,721$, and $|U| = 1,608,081$.

\noindent \textbf{Facebook Sales Conversion Dataset}
Facebook Sales Conversion Dataset \cite{fbdata} is a dataset used for Facebook ad campaign analysis and sales prediction. Each row of the raw dataset contains information about age, gender, interest, and number of conversions from a unique user group. We treat each unique user group as a single user and distribute the conversions across $n$ days. Because Facebook is the only publisher, both FTA, LTA and UNI will give the same attribution weight. Key dataset statistics include $\text{GS} = 60$, $T = 3264$, and $|U| = 1143$.

\subsection{Experiments on Different Queries and Objective Functions}
In this section, we compare \adsmech with baseline methods across various datasets. For the Zipf, Normal, and Uniform datasets, we use LTA to generate the data vector $\datavec$. For the Criteo and Facebook datasets, we generate three versions of the data vector $\datavec$ using different attribution methods: LTA, FTA, and UNI. We set the number of days to $n=31$ and the privacy budget to $\rho = 1$. 

In \adsmech and \mmbpc,  we use the differentially private quantile algorithm for the first $l=7$ days, the quantile percentage is $p=99\%$. In the SVT algorithm, the scale parameters are $\gap^{\uparrow} = 1.3$ and $\gap^{\downarrow}=0.8$, the thresholds are $T^{\uparrow} = T^{\downarrow} = 50$, the maximum number of reports is $k=7$.
For the allocation of privacy budget, we set $\rho_1 = 0.7 \rho$, $\rho_2 = 0.15 \rho / l$ and $\rho_3 = 0.15 \rho$. {The experiment for time complexity is in Appendix \ref{app:time}. The experiment for parameter tuning is in Appendix \ref{app:tune}. The ablation study is in Appendix \ref{app:ablation}}


\noindent \textbf{Prefix Sum Query with Weighted Root Mean Squared Error.}
Under Scenario 1 in Section \ref{sec:scene}, table \ref{exp:prefix} presents the weighted root mean squared error (WRMSE) of various methods across different datasets. \adsmech consistently outperforms other baseline methods. The final column displays the relative improvement over the second-best method; \adsmech achieves 33\% to 63\% lower error.

\noindent \textbf{Sliding Window Sum Query with Max Var Error.}
Under Scenario 2 in Section \ref{sec:scene}, table \ref{exp:sliding} presents the max variance error of various methods across different datasets. \adsmech consistently outperforms other baseline methods. The final column displays the relative improvement over the second-best method; \adsmech achieves 66\% to 95\% lower error.

We can see that \adsmech outperforms the global sensitivity-based methods \idglobal, \bintree, and \matglobal. The reason for this is that global sensitivity relies on the maximum number of contributions a user makes within the dataset, which is often significantly higher than the average contribution per user. As a result, these methods introduce an excessive amount of noise, leading to higher errors. In contrast, \adsmech clips user contributions on a daily basis, effectively reducing the sensitivity and minimizing the added noise. \bincrop performs poorly as it is primarily designed for continuous releases exceeding $10^5$ events. In contrast, \adsmech is better suited for advertising measurement, where the number of online releases is comparatively small. While \mmbpc employs the same approach to calculate the per-day user contribution bound and utilizes a complex matrix factorization technique, its performance still falls short of \adsmech. This highlights that \adsmech is effective despite its simple design.

\subsection{Experiments on Different Parameters}
In this section, we demonstrate how different parameters influence the errors. We use the Criteo dataset with LTA for illustration.

\noindent \textbf{Experiment on Different Privacy Budgets.}
In this section, we study the performance of different algorithms with respect to the privacy budget $\rho$, $\rho$ is chosen from [0.25, 0.5, 1, 2, 4, 8]. The dataset we use is the Criteo dataset with LTA as described in section \ref{sec:dataset}. The number of days is $n=31$. As shown in Figure \ref{fig:criteo-rho}, the errors associated with various algorithms decrease as $\rho$ increases. This is because the amount of noise added to the dataset decreases with a larger privacy budget, leading to lower errors.

\noindent \textbf{Experiments on Different Ad Campaign Durations.}
In this section, we study the performance of different algorithms with respect to the number of days $n$, $n$ is chosen from [7, 15, 31, 63, 127]. The dataset we use is the Criteo dataset with LTA. The total privacy budget is set as $\rho=1$. Figure \ref{fig:criteo-day} illustrates the performance of various algorithms. \matmech outperforms others consistently. {The experiment for $n=365$ is in Appendix \ref{app:365}.}

\section{Limitations and Future Work}\label{sec:limit}
{

Our approach currently assumes the presence of a trusted third party, similar to existing multi-touch attribution (MTA) systems where entities like advertising service providers (e.g., Criteo) or mobile measurement partners fulfill this role. However, reliance on a centralized third party introduces privacy risks, as some have faced regulatory scrutiny for mishandling user data. While our proposed algorithm with differential privacy (DP) mitigates inference attacks and cross-tracking risks, it does not fully eliminate concerns if a third party acts maliciously. As future work, we plan to explore decentralized alternatives to reduce or eliminate this dependency. Techniques such as Trusted Execution Environments (TEE) can shift trust from third parties to secure hardware providers, ensuring that only approved computations are executed. Alternatively, Multi-Party Computation (MPC) offers mathematically proven security without requiring a trusted intermediary. However, practical challenges remain in designing efficient protocols and ensuring DP guarantees within an MPC framework. Addressing these challenges will be a key direction for future research.

}

\section{Conclusions}

Our work is the first to formulate the real-world problem of MTA advertising measurement systems with real-time reporting of streaming data in advertising campaigns.
We introduced AdsBPC, a novel differentially private solution for streaming advertising measurement results. 
Extensive experiments confirm that AdsBPC significantly enhances measurement accuracy by \addcolor{33\% to 95\%} compared to existing streaming differential privacy mechanisms applied in advertising measurement.
This improvement underscores the effectiveness of our approach in achieving superior accuracy while maintaining a formal privacy guarantee, thus advancing the state-of-the-art in privacy-preserving advertising measurement.





\balance
\bibliographystyle{acm}
\bibliography{ref.bib}


\appendices


\section{Additional Experiments}

{


In this section, we present additional experiments to further demonstrate the efficiency and practicality of AdsBPC. All experiments follow the same settings as in Table \ref{exp:prefix}. We use Criteo (with LTA) and Facebook (FB) datasets.

\subsection{Scalability}\label{app:time}

AdsBPC can be applied as a post-processing step after the original data has been processed. Specifically, we require the accumulated data—namely, the total number of conversions for a given user, advertiser, and publisher on a particular day. This statistic can be efficiently retrieved using an SQL query in large databases with limited overhead. Regarding the computational cost of our algorithm, Table \ref{exp:time} shows the running time (in seconds) for different algorithm. The primary computational overhead of AdsBPC arises from tracking and truncating user contributions across different publishers. As future work, we will explore potential approaches to mitigate this overhead, as we believe an efficient implementation can significantly improve runtime performance. The running time on Facebook (FB) dataset is small because it contains only 1 publisher.
\begin{table}[h]
\centering
\begin{tabular}{|c|r|r|r|r|r|r|r|}
\hline
Dataset & IPA & \bintree & \bincrop & UMM & MMBPC &AdsBPC \\
\hline
Criteo    & 27 & 33 & 208 & 91 & 155 & 112  \\ 
FB    & 0.001 & 0.08 & 0.03 & 0.21 &0.43 & 0.004  \\ 
\hline
\end{tabular}
\caption{The running time (in seconds) comparison.}
\label{exp:time}
\end{table}

\subsection{n=365 Experiment}\label{app:365}

We evaluate the performance of AdsBPC for long-duration campaigns with n=365. The Criteo dataset is derived from real-world data collected over a 30-day window. If we simply divide it into 365 days, the data distribution changes significantly—for instance, the number of conversions per user per day would be diluted to approximately 1/10 of the original, leading to a lower signal-to-noise ratio and increased sparsity. To maintain a distribution similar to real-world data, we replicate the dataset five times to generate a new dataset with a higher signal strength that better reflects real-world conditions. We apply the same approach to the FB dataset. As shown in Table \ref{exp:n365}, AdsBPC consistently outperforms baseline methods, demonstrating its capability to handle long-duration campaigns effectively.

\begin{table}[h]
\centering
\begin{tabular}{|c|r|r|r|r|r|r|r|}
\hline
Dataset & IPA & \bintree & \bincrop & UMM & MMBPC &AdsBPC \\
\hline
Criteo   & 966 & 353 & 605 & 249 & 500 & \textbf{128} \\ 
FB      & 1206 & 548 & 784 & 278 & 257 & \textbf{98} \\ 
\hline
\end{tabular}
\caption{Error comparison when $n=365$.}
\label{exp:n365}
\end{table}

\subsection{Parameter Tuning}\label{app:tune}


In the experiment, we set the privacy budget allocation as $\rho_1 = 0.7 \rho$, $\rho_2 = 0.15 \rho / l$, and $\rho_3 = 0.15 \rho$. We examine how different privacy budget allocations affect the performance of AdsBPC. Specifically, we allocate $\rho_1 = (1-\lambda) \rho$ for measurement, $\rho_2 = \frac{\lambda}{2 l} \rho$ for DP-quantile, and $\rho_3 = \frac{\lambda}{2} \rho$ for SVT. Table \ref{exp:ratio} presents the performance of AdsBPC for different values of $\lambda$. From the results, we observe that when $\lambda=0.5$, an insufficient budget is allocated for measurement, leading to higher variance. Conversely, when $\lambda=0.1$, the bound selection becomes less accurate, resulting in increased bias.

\begin{table}[h]
\centering
\begin{tabular}{|c|r|r|r|r|r|r|}
\hline
$\lambda$  & 0.5 & 0.4 & 0.3 & 0.2 & 0.1  \\
\hline
Criteo   & 24.73 & 18.55 & 18.23 & 22.97 & 24.18 \\ 
FB      & 29.67 & 11.17 & 9.78 & 10.25& 18.40 \\ 
\hline
\end{tabular}
\caption{Different budget allocations.}
\label{exp:ratio}
\end{table}

Next, we examine the effect of the scaling factors $s^{\uparrow}, s^{\downarrow}$ used in SVT. We vary $s^{\uparrow}$ and set $s^{\downarrow} = 1/s^{\uparrow}$ accordingly. Table \ref{exp:scale} shows that the performance remains relatively robust against different selections of $s^{\uparrow}$.

\begin{table}[h]
\centering
\begin{tabular}{|c|r|r|r|r|r|r|}
\hline
$s^{\uparrow}$ & 1.1 & 1.2 & 1.3 & 1.4 & 1.5  \\
\hline
Criteo   & 21.89 & 22.39 & 18.23 &20.45 & 18.89 \\ 
FB      & 12.23 & 11.73 & 9.78 & 10.70& 12.05 \\ 
\hline
\end{tabular}
\caption{Different scaling factors.}
\label{exp:scale}
\end{table}

Lastly, we examine the impact of the thresholds $T^{\uparrow}, T^{\downarrow}$ used in SVT. We vary $T^{\uparrow}$ and set $T^{\downarrow} = T^{\uparrow}$. Table \ref{exp:svt} shows that a smaller threshold $T^{\uparrow}$ causes SVT to update more frequently, leading to faster budget consumption and leaving the contribution bound unchanged in later days. This results in higher errors due to inaccurate estimates of contribution bounds in later days.

\begin{table}[h]
\centering
\begin{tabular}{|c|r|r|r|r|r|r|}
\hline
$T^{\uparrow}$ & 10 & 25 & 50 & 75 & 100    \\
\hline
Criteo  & 23.37& 21.50 & 18.23 & 20.89 &  21.45  \\ 
FB    & 17.79 &  16.90 & 9.78 & 8.71 & 10.13  \\ 
\hline
\end{tabular}
\caption{Different SVT thresholds.}
\label{exp:svt}
\end{table}

\subsection{Ablation Study}\label{app:ablation}
We aim to investigate the following questions:

1. \textbf{Constant Bound}: What if we use a fixed bound instead of DP-quantile and SVT? How would this impact accuracy? We use the same setting as in Table \ref{exp:prefix} ($\rho=1$) but exclude dp-quantile and SVT, relying solely on a predefined user contribution bound for all 31 days. The results in Table \ref{tab:ab1} indicate that WRMSE is highly sensitive to this bound: a small bound ($r=1$) introduces high bias, while a larger bound ($r=6$) increases variance. Although choosing an intermediate bound ($r=3$) can mitigate error, the optimal bound for a given dataset is itself private information and unknown in advance. This highlights the necessity of a private estimation method for the user contribution bound.

\begin{table}[h]
    \centering
    \begin{tabular}{|c|c|c|c|c|c|c|c|}
        \hline
       Bound  & r=1 & r=2 & r=3 & r=4 & r= 5 & r = 6 \\ \hline
        Criteo & 1512.25 & 24.52 &  18.18& 22.29 & 23.15& 40.63\\ \hline
        FB  & 230.48 & 53.25 &  9.99& 16.35 & 26.16& 30.68\\ \hline
    \end{tabular}
    \caption{Ablation Study 1: Constant Contribution Bound}
    \label{tab:ab1}
\end{table}

2. \textbf{Daily DP-Quantile Updates}: What if we apply DP-quantile every day without SVT? How would this affect accuracy? We use the same setting as in Table \ref{exp:prefix} but with a different privacy budget. Instead of SVT, we use dp-quantile to estimate user contributions over all 31 days. Table \ref{tab:ab2} shows the errors. The performance deteriorates due to a lower privacy budget allocated per day, leading to a less accurate estimation of the user contribution bound. However, as the privacy budget increases, performance improves, highlighting the importance of accurately estimating the user contribution bound.

\begin{table}[h]
    \centering
    \begin{tabular}{|c|c|c|c|c|c|c|c|}
        \hline
       Budget  & $\rho=1$ & $\rho=2$ & $\rho=4$ & $\rho=6$ & $\rho=8$ & $\rho=10$  \\ \hline
        Criteo & 79.08 & 72.02 & 55.83 & 42.41& 33.45 & 32.68 \\ \hline
        FB &137.55  & 100.82 & 65.64 & 33.24 & 29.34 & 29.11 \\ \hline
    \end{tabular}
    \caption{Ablation Study 2: DP-Quantile Only, No SVT}
    \label{tab:ab2}
\end{table}

}

\section{Possibly Optimality of AdsBPC}
\label{app:opt}

To get the optimal strategy matrix under user level differential privacy with bounded per-day contribution, we use the numerical solver in \cite{choquette2022multi} under its default parameters. For each $n$, we randomly generate the per-day constraint $r_i \in [1, 10]$ and set $\rho=1$. Surprisingly, we find that the optimal strategy matrix $\mat{B}^{\star}$ tends to be a diagonal matrix. Table \ref{tab:diagonal} shows the absolute value of the off-diagonal elements in $\mat{P}^{\star}$, the maximum absolute value of $\mat{B}^{\star}$ is close to 0. We also compare the optimal objective function values of the solver and AdsBPC. Table \ref{tab:obj} shows that our method achieves optimality up to a difference of $10^{-2}$, which may come from numerical error.

\begin{table}[h]
\centering
\caption{Check the Absolute Value of the Off-Diagonal Elements.}
\begin{tabular}{|c|c|c|c|}
\hline
$n$ & $\max_{i \neq j} |\mat{B}^{\star}[i,j]|$ &  $\min_{i \neq j} |\mat{B}^{\star}[i,j]|$  &  $\operatorname{mean}_{i \neq j} |\mat{B}^{\star}[i,j]|$ \\
\hline
2 & $6.67 * 10^{-5}$ &  $6.67 * 10^{-5}$ & $6.67 * 10^{-5}$ \\
4 & $1.59*10^{-4}$ & $1.14*10^{-5}$ & $5.38*10^{-5}$ \\
6 &  $7.26 * 10^{-5}$ &$8.03 * 10^{-7}$ & $1.82 * 10^{-5}$ \\
8 & $4.81 * 10^{-5}$ & $4.70 * 10^{-8}$ & $5.58 * 10^{-6}$\\
10 & $4.39 * 10^{-5}$ & $1.69 * 10^{-8}$& $3.03 * 10^{-6}$\\
12 & $4.36 * 10^{-5}$ & $9.20 * 10^{-9}$& $2.14 * 10^{-6}$ \\
14 &$4.48 * 10^{-5}$  & $2.79 * 10^{-10}$ & $1.71 * 10^{-6}$\\
\hline
\end{tabular}
\label{tab:diagonal}
\end{table}

\begin{table}[h]
\centering
\caption{Comparison between the optimal objective function value of different methods.}
\begin{minipage}{0.23\textwidth}
\centering
\begin{tabular}{|c|c|c|}
\hline
$n$ & Solver & AdsBPC\\
\hline
2 &  108.46 &  108.46\\
4  & 950.55 &  950.56\\
6  & 2400.11 &  2400.12\\
8  & 5193.02  & 5193.03 \\
\hline
\end{tabular}
\end{minipage}
\hfill
\begin{minipage}{0.23\textwidth}
\centering
\begin{tabular}{|c|c|c|}
\hline
$n$ & Solver & AdsBPC\\
\hline
10  & 11365.32 & 11365.32 \\
12  &  19441.07&  19441.07\\
14  &  32607.31 &  32607.32\\
\hline
\end{tabular}
\end{minipage}

\label{tab:obj}
\end{table}

\section{Proofs}

\subsection{Proof for Theorem \ref{thm:sen}}

\label{app:sensitivity}


\begin{proof}
    From the definition of sensitivity \ref{def:l2sen}, the function $f(\datavec) = \mat{B} \datavec$ has the sensitivity $ \max_{\vec{d} \in \subperdaydiff} \| \mat{B} \vec{d}\|_2 $. Notice that $\|\mat{B} \vec{d}\|_2$ is a convex and continuous function with respect to $\vec{d}$, and the constraint set $\subperdaydiff$ is convex and compact with respect to $\vec{d}$. So according to Bauer's maximum principle \cite{kruvzik2000bauer}, the function $\| \mat{B} \vec{d} \|_2$ attains its maximum at some extreme point of the set $\subperdaydiff$. Recall that $\subperdaydiff = \{\vec{d} \in \mathbb{R}^n: -r_i \leq \vec{d}_i \leq r_i, 1\leq i\leq n\}$, so the set of extreme points of $\subperdaydiff$ is 
\begin{align}
    \label{eq:diffex}
    \diffex = \{\vec{d} \in \mathbb{R}^n : d_i \in \{-r_i, r_i \}, 1\leq i\leq n \}.
\end{align}

Then the square of sensitivity of $\mat{B}$ is,
\begin{align*}
    \max _{\mathbf{d} \in \subperdaydiff} \mathbf{d}^{T} \mat{B}^T \mat{B} \mathbf{d} = \max_{\vec{d} \in \diffex} \mathbf{d}^{T} \mat{B}^T \mat{B} \mathbf{d} = \frac{r_1^2}{\sigma_1^2} + \cdots + \frac{r_n^2}{\sigma_n^2} .
\end{align*}
\end{proof}

\subsection{Proof for Theorem \ref{thm:dp}}
\label{app:zCDP}


\begin{proof}
    Let $f(\datavec) = \mat{B} \datavec$, from theorem \ref{thm:sen} the sensitivity of $f$ is $\Delta_2(f) = \sqrt{\frac{r_1^2}{\sigma_1^2} + \cdots + \frac{r_n^2}{\sigma_n^2} }$. The constraint in problem \ref{eq:cauchy} makes sure that $\Delta_2^2(f) \leq 2\rho$. From theorem \ref{thm:gaussian} we get $f(\datavec) + N(0, \mat{I}_n) = \mat{B}\datavec + \vec{z}$ satisfies $\rho$-zCDP. Using the post-processing property \ref{thm:post}, we know that $\mech(\datavec)$ = $\mat{L} (\mat{B} \datavec + \mat{z}) = \mat{W} \datavec+ \mat{L}\mat{z} $ satisfies $\rho$-zCDP.
\end{proof}

\subsection{Proof for Theorem \ref{thm:quantile}}


\begin{proof}
From Theorem \ref{thm:exp}, the Exponential Mechanism satisfies $\epsilon$-DP and $\frac{1}{8} \epsilon^2$-zCDP. From Theorem \ref{thm:puredp}, any $\epsilon$-DP mechanism is also $\frac{e^{\epsilon} - 1}{e^{\epsilon} + 1}\epsilon$-DP. As a special case of Exponential Mechanism, the PrivateQuantile algorithm satisfies $\rho_2$-zCDP, with $\rho_2=\min(\frac{1}{8} \epsilon^2, \frac{e^{\epsilon} - 1}{e^{\epsilon} + 1} \epsilon)$.
\end{proof}

\subsection{Proof for Theorem \ref{thm:svt}}

\begin{lemma}
\label{lemma:l1sen}
     The $L_1$ sensitivity of the following is 1.
     \begin{align*}
         q(D, \tau, \gap) = \abovet(D, \tau) - \abovet(D, \tau *\gap).
     \end{align*}
    Here $\gap <1$.
\end{lemma}

\begin{proof}
Let $D$ and $D'$ be two adjacent datasets. $D'$ can be obtained from $D$ by replacing one user $u$ in $D$ with another user $u'$. $q(D, \tau, \gap)$ counts the number of people in $D$ with contributions between $(\tau * \gap, \tau]$. If both $u$ and $u'$ have number of contributions smaller or equal to $\tau * \gap$, then $q(D, \tau, \gap) = q(D', \tau, \gap)$. If both $u$ and $u'$ have number of contributions between $(\tau + \gap, \tau]$, then $q(D, \tau, \gap) = q(D', \tau, \gap)$. If both $u$ and $u'$ have number of contributions larger than $\tau$, then $q(D, \tau, \gap) = q(D', \tau, \gap)$. If one of them (suppose it's $u$) has number of contributions smaller or equal to $\tau$, and another one (suppose it's $u'$) has number of contributions larger than $\tau$, then 
\begin{align*}
     &q(D', \tau, \gap) \\
     =& \abovet(D', \tau) - \abovet(D', \tau *\gap)\\
     =&\abovet(D, \tau) + 1 - (\abovet(D, \tau *\gap) +1 ) \\
     =& \abovet(D, \tau) - \abovet(D, \tau *\gap)
     = q(D, \tau, \gap).
\end{align*}

If one of them has number of contributions between $(\tau, \tau + \gap]$, and another one doesn't have number of contributions between $(\tau, \tau + \gap]$, then $\|q(D, \tau, \gap) - q(D', \tau, \gap) \|_1= 1$. We see that the $L_1$ difference is at most 1. This is true for the same reason when $\gap < 0$. Therefore, the $L_1$ sensitivity is 1.
    
\end{proof}


With Lemma \ref{lemma:l1sen}, we are ready to prove Theorem \ref{thm:svt}.
\begin{proof}
The sensitivity of $q^{\uparrow} = \abovet(D, \tau)$ is 1 because if we replace a user in $D$, the total number of users who have contributions greater than $\tau$ can change by at most 1. From Lemma \ref{lemma:l1sen} we know that the sensitivity of $q^{\downarrow}$ is 1. Therefore, each SVT satisfies $\epsilon / 2$-DP, so UpdateBoundSVT satisfies $\epsilon$-DP. From Theorem \ref{thm:puredp}, UpdateBoundSVT satisfies $\rho_3 = \frac{e^{\epsilon} - 1}{e^{\epsilon} + 1} \epsilon$-zCDP.
\end{proof}

\subsection{Proof for Theorem \ref{thm:adsbpc}}


\begin{proof}
In Algorithm \ref{alg:adsbpc}, we allocate a privacy budget of $\rho_2$ for each call to PrivateQuantile. Since there are $l$ calls in total, the overall budget for PrivateQuantile is $l * \rho_2$. Additionally, we use a privacy budget of $\rho_3$ for UpdateBoundSVT. From Theorem \ref{thm:adsbpc}, recall that the sensitivity of the noise addition mechanism is determined by $\frac{r_i}{\sigma_i}$. Since we update the user contribution bound $r_i$ by scaling $\sigma_i$ proportionally, the sensitivity remains unchanged. Therefore, the noise addition mechanism uses a privacy budget of $\rho_1$. In total, the combined privacy budget used is $\rho = \rho_1 + l \cdot \rho_2 + \rho_3$.
\end{proof}

\subsection{Proof for Theorem \ref{thm:sen-pub}}

\begin{proof}
According to the matrix value mechanism \cite{chanyaswad2018mvg}, the function $f(\mat{X}) = \mat{B} \mat{X}$ has the sensitivity $ \max_{\vec{D} \in \subperdaydiff} \| \mat{B} \vec{D}\|_F $. Recall that $\subperdaydiff = \{\vec{D} \in \mathbb{R}^{n\times k}: 0 \leq \|\vec{D}[i, :]\|_1 \leq r_i, 1\leq i\leq n\}$, then the square of sensitivity of $\mat{B}$ is,
\begin{align*}
    \max _{\mathbf{D} \in \subperdaydiff} \sum_{1\leq i \leq n} \frac{\|\mat{D}[i, :] \|_2^2}{\sigma_i^2} &\leq \sum_{1\leq i \leq n} \frac{\|\mat{D}[i, :] \|_1^2}{\sigma_i^2} 
    \leq \sum_{1 \leq i \leq n} \frac{r_i^2}{\sigma_i^2}.
\end{align*}
\end{proof}

\subsection{Proof for Theorem \ref{thm:dp-pub}}

\begin{proof}
    Let $f(\mat{X}) = \mat{B}\mat{X}$, from theorem \ref{thm:sen-pub} the sensitivity of $f$ is $\Delta_2(f) = \sqrt{\frac{r_1^2}{\sigma_1^2} + \cdots + \frac{r_n^2}{\sigma_n^2} }$. The constraint in problem \ref{eq:cauchy} makes sure that $\Delta_2^2(f) \leq 2\rho$. We can flatten $f(\mat{X})$ and $\mat{Z} \sim MVG_{n, k}(0, \mat{I}_n, \mat{I}_k)$ as vectors of length $n \times k$. From Theorem \ref{thm:gaussian}, we get $f(\mat{X}) + \mat{Z} = \mat{B}\mat{X} + \vec{Z}$ satisfies $\rho$-zCDP. Using the post-processing property, we know that $\mech(\mat{X})$ = $\mat{L} (\mat{B} \mat{X} + \mat{Z}) = \mat{W} \mat{X}+ \mat{L}\mat{Z} $ satisfies $\rho$-zCDP.
\end{proof}

\section{PrivateQuantile Algorithm}
Algorithm \ref{alg:quantile} shows the pseudo-code for the DP-Quantile algorithm \cite{gillenwater2021differentially, smith2011privacy}. Note that PrivateQuantile requires an additional parameter, $\Lambda$, which serves as an upper bound. In our experiments, we set $\Lambda = 10$ for the Crieto, Facebook, and Zipf datasets, and $\Lambda = 20$ for the Normal and Uniform datasets. \cite{durfee2023unbounded} proposed an algorithm that does not require such a bound. Notably, their algorithm can be integrated into our system with minimal modifications.

\begin{algorithm}
\caption{PrivateQuantile}
\label{alg:quantile}
\begin{algorithmic}[1]
\STATE \textbf{Input:} Dataset $D$; Quantile percentage $p$; privacy parameter $\epsilon$, bounding parameter $\Lambda$.
\STATE \textbf{Output:} The selected user contribution bound based on $p$ quantile.

\STATE $X$ = [].

\FOR {Each user $u_i$ in $D$}
\STATE $c_i \gets $ the number of records from $u_i$.
\STATE $X$.append($c_i$).
\ENDFOR

\STATE Sort $X$ in ascending order. $X=[c_1, c_2, \cdots, c_k]$.
\STATE Replace $c_i > \Lambda$ with $\Lambda$. Define $c_0 = 0$ and $c_{k+1}$.

\FOR { $i=1$ to $k$ }
\STATE Set $y_i = (c_{i+1} - c_i) \exp(-\epsilon |i - p *k|/2)$.
\ENDFOR 
\STATE Sample an integer $i \in \{ 0, \cdots, k \}$ with probability $y_i / (\sum_{i=0}^k y_i)$.
\STATE $c \gets$ uniformly sampled from $[c_i, c_{i+1}]$.
\STATE Return $c$.
\end{algorithmic}
\end{algorithm}

\newpage 

\section{Meta-Review}

The following meta-review was prepared by the program committee for the 2025 IEEE Symposium on Security and Privacy (S\&P) as part of the review process as detailed in the call for papers.

\subsection{Summary} This paper presents AdsBPC, a system to measure ad campaign exposure and attribution with differential privacy guarantees. This system allows publishers to measure multi-touch attribution (MTA) while providing user-level differential privacy (DP). The authors claim a significant improvement in accuracy when compared to prior work.

\subsection{Scientific Contributions} \begin{itemize} \item Provides a Valuable Step Forward in an Established Field \item Creates a New Tool to Enable Future Science \end{itemize}

\subsection{Reasons for Acceptance} \begin{enumerate} \item The paper provides a valuable step forward in an established field. AdsBPC novel approach has been evaluated on five datasets, both synthetic and real-world, and achieves improvements over baselines such as IPA, BIN, Stream and UMM. Furthermore, the feasibility and performance of the system was tested across ad campaigns with different duration of up to 365 days.

\item The reviewers were positive about the presentation of the paper, agreeing that it does a really good job of presenting the methodology behind AdsBPC and the thorough comparison with existing technologies. The reviewers also agreed that, while some of the underlying techniques and notions of the paper might not be novel, the way that these are applied to the specific use case of advertisements and measuring conversions provides enough of a step forward on the field to justify accepting this paper.

\end{enumerate}

\end{document}